\documentclass[11pt]{amsart}
\usepackage[english]{babel}
\usepackage{amssymb}
\usepackage{amsmath}

\usepackage{lscape}

\newtheorem{dummy}{Dummy}

\newtheorem{lemma}[dummy]{Lemma}

\theoremstyle{definition}
\newtheorem{definition}{Definition}

\newtheorem{example}[dummy]{Example}

\newtheorem{remark}[dummy]{Remark}

\newcommand{\ignore}[1]{}

\author{S. Pumpl\"un}
\email{susanne.pumpluen@nottingham.ac.uk}
\address{School of Mathematical Sciences\\
University of Nottingham\\
University Park\\
Nottingham NG7 2RD\\
United Kingdom
}


\begin{document}

\title[How to obtain lattices from $(f,\sigma,\delta)$-codes ]
{How to obtain lattices from $(f,\sigma,\delta)$-codes via a generalization of Construction A}

\begin{abstract}
We show how cyclic $(f,\sigma,\delta)$-codes over finite rings canonically induce a $\mathbb{Z}$-lattice in
$\mathbb{R}^N$ by using certain quotients of orders in nonassociative division algebras defined using the skew polynomial
$f$.  This construction generalizes the one using certain
 $\sigma$-constacyclic codes by Ducoat and Oggier, which used quotients of orders in non-commutative
 associative division algebras defined by $f$, and can be viewed as a generalization of the classical Construction A for lattices from linear codes.
It has the potential to be applied to coset coding, in particular to wire-tap coding.
Previous results by Ducoat and Oggier are obtained as special cases.
\end{abstract}

\maketitle

%*******************************************************************************************%
%
%Introduction
%
%*******************************************************************************************%

\section*{Introduction}

In the classical Construction A, a lattice is obtained by lifting a linear
code over some finite ring \cite{CS}. This idea was recently generalized to the non-commutative setting
by considering natural orders in  cyclic algebras over number fields: by taking the quotient of
the natural order by a suitable ideal, a ring is obtained which is isomorphic
to the quotient of a twisted polynomial ring by some polynomial  \cite{DO, OS}. This
established a connection between twisted polynomials and certain $\sigma$-constacyclic codes.

We generalize Construction A using skew polynomial rings $S[t;\sigma,\delta]$ and  construct lattices
by lifting cyclic $(f,\sigma,\delta)$-codes, i.e. much more general linear codes than considered in \cite{DO, OS}, to lattices in nonassociative algebras.
The multiplicative structure of the algebra is not necessary to build a lattice, so we do not limit our considerations
to associative algebras as has been done so far.

 As
recently several classes of cyclic $(f,\sigma,\delta)$-codes were constructed with a better minimal
distance for certain lengths than previously known codes (e.g., see  \cite{B},
\cite{BG}, \cite{BSU08},  \cite{BU14}, \cite{BU14.2}, \cite{BU09}, \cite{BU09.2} \cite{Cao},  \cite{FG},
  \cite{GK}, \cite{LL},  \cite{MW}),
 $(f,\sigma,\delta)$-codes become increasingly important. These codes  employ skew polynomial rings
$S[t;\sigma,\delta]$  where $S$ is a unital ring,
$\sigma$ an injective endomorphism of $S$ and $\delta$ a left $\sigma$-derivation of $S$,
and are  built by
choosing a  monic polynomial $f\in S[t;\sigma,\delta]$ of degree $m$, and some monic right divisor $g$ of $f$
\cite{BL13}.
Every cyclic $(f,\sigma,\delta)$-code is associated with a  principal left ideal
of a unital nonassociative algebra $S_f$ defined by $f$, which is generated by some monic right divisor $g$ of $f$.

The nonassociative algebra
$S_f=S[t;\sigma,\delta]/S[t;\sigma,\delta]f$ is defined on the additive subgroup
$\{h\in S[t;\sigma,\delta]\,|\, {\rm deg}(h)<m \}$ of $S[t;\sigma,\delta]$ by using right division by $f$
to define the algebra multiplication $g\circ h=gh \,\,{\rm mod}_r f $  \cite{P15}.
This can be seen as a canonical generalization of associative quotient algebras
$S[t;\sigma,\delta]/(f)$, where we factor out a two-sided ideal generated by $f$, which occurs when $Rf$ is a two-sided
ideal.
 If $S$ is a division algebra, the associative
quotient algebras $S[t;\sigma,\delta]/(f)$ as well as the right nuclei of the
nonassociative algebras
$S_f$ were used when constructing central simple algebras for instance in  \cite{Am}, \cite{Am2},  \cite{Hoe},
\cite[Sections 1.5, 1.8, 1.9]{J96},  \cite{O}. Due to their large nuclei, the algebras $S_f$ were also  successfully employed to systematically build
fast-decodable fully diverse space-time block codes in \cite{SP14, R13, MO13}, see \cite{PS15},
which are used for reliable high rate transmission
over wireless digital channels with multiple antennas transmitting and receiving the data.
Skew-polynomial rings and their ideals have been already used in other applications and when generalizing other classical notions like
Gr\"{o}bner bases \cite{M0} to a non-commutative setting, e.g. see
\cite{BGV, CM, Lev, Lev2, KrW, Kr, P1, P2, Wei}, where they appear as examples of solvable polynomial rings, operator theory
\cite{G6}, and other codes, in particular (cyclic) convolutional codes and MDS codes cf. \cite{G1, G2, G4, G5, N1, N2, N3}.

We choose suitable monic irreducible skew polynomials
$f\in K[t,\sigma,\delta]$ with
$K/F$ a finite field extension of number fields, or $f\in D[t,\sigma,\delta]$ with $D$ a cyclic division algebra over
 a number field, and define natural
orders $\Lambda$ in $S_f$. We then use the quotient of $\Lambda$ by certain
two-sided ideals  to canonically construct a
lattice $L$ in $\mathbb{R}^N$, i.e. a $\mathbb{Z}$-module $L$ of rank $N$,
 from a cyclic $(f,\sigma,\delta)$-code over a finite ring.

The non-commutative setup  treated in \cite{DO, OS} is obtained as the special case where
$K/F$ is a cyclic field extension of degree $n$ and $f(t)=t^n-c\in \mathcal{O}_F[t;\sigma]$ is
\emph{(right-)invariant}, i.e. satisfy
 $fR\subset Rf$, which makes $Rf$ a two-sided ideal, and $S_f$ non-commutative, but still associative.

The advantage of using  nonassociative algebras as we do is the fact that this does not limit our choices of skew
polynomials $f$ to those which create two-sided ideals $Rf$. This means that we have a much larger choice of lattices
we can build. Lattices now can be obtained
by lifting  any cyclic $(f,\sigma,\delta)$-code, moreover, we can also lift
$\sigma$-constacyclic codes to lattices (now sitting inside nonassociative algebras).
Sometimes there exist easy conditions for nonassociative cyclic algebras to be division algebras which is an additional bonus.

Our construction A  can be used  to encode
space-time block codes, for coset coding, and in particular for wiretap coding.

The paper is organized as follows:
After collecting  the results we need in Section 1,  for monic and irreducible
$f\in K[t;\sigma,\delta]$ we define a natural order in $S_f$, and investigate the quotients of a natural order
 by some ideals in Section \ref{sec:naturalI}. These results are then generalized  in Section \ref{sec:naturalII} to
 monic irreducible  $f\in D[t;\sigma,\delta]$, where $D=(K/F,\rho,c)$ is a cyclic division algebra.
In Sections \ref{sec:codes} and \ref{sec:codesII}, we describe a
  lattice encoding of certain cyclic $(f,\sigma,\delta)$-codes over the finite rings $\mathcal{O}_K/\mathfrak{p}\mathcal{O}_K$,
 where $\mathfrak{p}$ is a  maximal ideal in some suitable subring of $\mathcal{O}_K$, and how it can be applied to space-time block codes.

Throughout the paper we will put  a special emphasis on the
 nonassociative cyclic algebras $(K/F,\sigma,c)$ employed in \cite{SPO12}, and on the generalized nonassociative cyclic algebras
 $(D, \sigma, d)$, since these are used for iterated space-time block codes \cite{SP14}, \cite{PS15}.

%*******************************************************************************************%
%
%Preliminaries
%
%*******************************************************************************************%

\section{Preliminaries}

\subsection{Nonassociative algebras}
Let $R$ be a unital commutative ring and let $A$ be an
$R$-module.
We call $A$ an \emph{algebra} over $R$ if there exists an
$R$-bilinear map $A\times A\to A$, $(x,y) \mapsto x \cdot y$, denoted simply by juxtaposition $xy$,
the  \emph{multiplication} of $A$.
An algebra $A$ is called \emph{unital} if there is
an element in $A$, denoted by 1, such that $1x=x1=x$ for all $x\in A$.
We will only consider unital algebras.

For an $R$-algebra $A$,
the {\it left nucleus} of $A$ is defined as ${\rm Nuc}_l(A) = \{ x \in A \, \vert \, [x, A, A]  = 0 \}$ where
 $[x, y, z] = (xy) z - x (yz)$ for $x,y,z\in A$, the
{\it middle nucleus}  as ${\rm Nuc}_m(A) = \{ x \in A \, \vert \, [A, x, A]  = 0 \}$ and  the
{\it right nucleus}  as ${\rm Nuc}_r(A) = \{ x \in A \, \vert \, [A,A, x]  = 0 \}$.
Their intersection
 ${\rm Nuc}(A) = \{ x \in A \, \vert \, [x, A, A] = [A, x, A] = [A,A, x] = 0 \}$ is the {\it nucleus} of $A$.
The {\it center} of $A$ is ${\rm C}(A)=\{x\in A\,|\, x\in \text{Nuc}(A) \text{ and }xy=yx \text{ for all }y\in A\}$
 \cite{Sch}.

Let $R$ be a Noetherian integral domain with quotient field $F$ and $A$ a finite-dimensional unital $F$-algebra.
Then an $R$-\emph{lattice} in $A$ is an $R$-submodule $\Gamma$ of $A$ which is
 finitely generated and contains an $F$-basis of $A$.
An $R$-\emph{order} in $A$ is a multiplicatively closed $R$-lattice containing $1_A$
(the multiplication may be not associative).
An $R$-order will be called \emph{maximal} if $\Gamma'\subset\Gamma$ implies $\Gamma'=\Gamma$ for every $R$-order
$\Gamma'$ in $A$.

   An algebra $A\not=0$ over a field $F$ is called a {\it division algebra}, if for any $a\in A$, $a\not=0$,
the right multiplication  with $a$, $L_a(x)=ax$,  and the right multiplication with $a$, $R_a(x)=xa$, are bijective.
Any division algebra is simple, that means has only trivial two-sided ideals.
A finite-dimensional algebra
$A$ is a division algebra over $F$ if and only if $A$ has no zero divisors.

%%%%%%%%%%%%%%%%%%%%%%%%%%%%%%%%%%%%%%%%%%%%%%%%%%%%%%%%%%%%%%%%%%%%%%%%%%%%%%%%%%%

\subsection{Skew polynomial rings}

%%%%%%%%%%%%%%%%%%%%%%%%%%%%%%%%%%%%%%%%%%%%%%%%%%%%%%%%%%%%%%%%%%%%%%%%%%%%%%%%%%

Let $S$ be a unital  (not necessarily commutative) ring, $\sigma$ an injective ring homomorphism of $S$ and
$\delta:S\rightarrow S$ a \emph{left $\sigma$-derivation}, i.e.
an additive map such that
$\delta(ab)=\sigma(a)\delta(b)+\delta(a)b$
for all $a,b\in S$, implying $\delta(1)=0$. Let ${\rm Const}(\delta)=\{a\in S\,|\, \delta(a)=0\}$ and
${\rm Fix}(\sigma)=\{a\in S\,|\, \sigma(a)=a\}$.

 The \emph{skew polynomial ring} $R=S[t;\sigma,\delta]$(defined first by Ore \cite{O1})
is the set of skew polynomials $a_0+a_1t$ $+\dots +a_nt^n$
with $a_i\in S$, where addition is defined term-wise and multiplication by
$ta=\sigma(a)t+\delta(a)$ for all $a\in S$
 (for properties see \cite{C, G3, G6}).
The ring $S[t;\sigma]=S[t;\sigma,0]$ is called a \emph{twisted polynomial ring} and $S[t;\delta]=S[t;id,\delta]$ a
 \emph{differential polynomial ring}.

 For $f=a_0+a_1t+\dots +a_nt^n$ with $a_n\not=0$ define ${\rm deg}(f)=n$ and ${\rm deg}(0)=-\infty$.
Then ${\rm deg}(fg)\leq{\rm deg} (f)+{\rm deg}(g)$
with equality if
\begin{itemize}
\item $f$ has an invertible leading coefficient,
\item $g$ has an invertible leading coefficient,
\item $S$ is a domain.
\end{itemize}
 An element $f\in R$ is \emph{irreducible} in $R$ if it is not a unit and it has no proper factors, i.e if there do not exist $g,h\in R$ with
 ${\rm deg}(g),{\rm deg} (h)<{\rm deg}(f)$ such that $f=gh$.

%%%%%%%%%%%%%%%%%%%%%%%%%%%%%%%%%%%%%%%%%%%%%%%%%%%%%%%%%%%%%%%%%%%%%%%

\subsection{How to obtain nonassociative algebras from skew polynomial rings}\label{sec:S_f}

%%%%%%%%%%%%%%%%%%%%%%%%%%%%%%%%%%%%%%%%%%%%%%%%%%%%%%%%%%%%%%%%%%%%%%%%%

From now on, let $R=S[t;\sigma,\delta]$ and $\sigma$ injective. We do not assume $S$ to be a division ring.
We can still perform a right division by a polynomial $f \in R$ which has invertible leading coefficient $d_m$:
for all $g(t)\in R$ of degree $l> m$,  there exist  uniquely determined $r(t),q(t)\in R$ with
 ${\rm deg}(r)<{\rm deg}(f)$, such that
$g(t)=q(t)f(t)+r(t).$
Let ${\rm mod}_r f$ denote the remainder of right division by such an $f$ \cite[Proposition 1]{P15}.

 Suppose $f(t)=\sum_{i=0}^{m}d_it^i\in R=S[t;\sigma,\delta]$ has an invertible leading coefficient $d_m$.
  Let  $R_m=\{g\in R\,|\, {\rm deg}(g)<m\}.$
Then $R_m$  together with the multiplication
$g\circ h=  gh \,\,{\rm mod}_r f$
becomes a  unital nonassociative ring $S_f=(R_m,\circ)$ also denoted by $R/Rf$  \cite{P15}.

This construction was  introduced by Petit \cite{P66,P68} for unital division rings $S$.
$S_f$ is a unital nonassociative algebra  over
 $S_0=\{a\in S\,|\, ah=ha \text{ for all } h\in S_f\}$ which is a commutative subring of $S$.
We call $S_f$ a \emph{Petit algebra}.
The algebra $S_f$ is associative if and only if  $Rf$ is a two-sided ideal in $R$
  (\cite[Theorem 4 (ii)]{P15}, or \cite[(1)]{P66} if $S$ is a division ring).
 For all invertible $a\in S$ we have $S_f\cong S_{af}$, so that without loss of generality it suffices to only
consider monic polynomials in the construction.

  If $S_f$ is not associative then $S\subset {\rm Nuc}_l(S_f)$ and $S\subset{\rm Nuc}_m(S_f)$,
${\rm Nuc}_r(S_f)=\{g\in R_m\,|\, fg\in Rf\}$ and $S_0$ is the center of $S_f$ \cite{P15}.
It is easy to see that
$C(S)\cap {\rm Fix}(\sigma)\cap {\rm Const}(\delta)\subset S_0.$

 If $S$ is a division algebra and $S_f$ is a finite-dimensional vector space over $S_0$,
  then $S_f$ is a division algebra if and only if $f(t)$ is irreducible in $R$ \cite[(9)]{P66}.

For $f(t)=\sum_{i=0}^{m}d_it^i\in S[t;\sigma]$, $t$ is left-invertible in $S_f$ if and only if $d_0$ is invertible
 by a simple degree argument. Thus if $f$ is irreducible (hence $d_0\not=0$) and $S$ a division ring then
$t$ is always left-invertible in $S_f$ and $S_0={\rm Fix}(\sigma)\cap C(S)$ is the center of $S_f$
\cite[Theorem 8 (ii)]{P15}.

The  $S$-basis $1, t, t^2, \ldots, t^{n-1}$ is  the \emph{canonical basis} for the left $S$-module $S_f$.
Since $ S\subset {\rm Nuc}_m(S_f)$ and $S\subset {\rm Nuc}_l(S_f)$, the right multiplication with $0\not=a\in S_f$ in $S_f$,
$R_h:S_f\longrightarrow S_f,$ $p\mapsto pa$, is an $S$-module endomorphism, and
after expressing  $R_a$ in matrix form
with respect to the canonical basis of $S_f$, the map
$$\gamma: S_f \to {\rm End}_K(S_f), a\mapsto R_a$$
induces an injective $S$-linear map
$$\gamma: S_f \to {\rm Mat}_m(S), a\mapsto R_a \mapsto M(a).$$
This fact is exploited when designing  space-time block codes which employ one of the following two
special cases of algebras:

\begin{definition}
(i) Let $S/S_0$ be an extension of commutative unital rings and $G=\langle \sigma\rangle$ a finite cyclic group of order
$m$ acting on $S$ such that $S_0={\rm Fix}(\sigma)$. For any $c\in S$,
$$S_f=S[t;\sigma]/S[t;\sigma] (t^m-c)$$
is called a \emph{nonassociative cyclic algebra}  $(S/S_0,\sigma,c)$ \emph{of degree $m$}.
\\ (ii)  Let $D$ be a finite-dimensional central division algebra
 over  $F={\rm C}(D)$ of degree $n$, $\sigma\in {\rm Aut}(D)$ such that
$\sigma|_{F}$ has finite order $m$ and $f(t)=t^m-d\in D[t;\sigma]$. Let $F_0=F\cap {\rm Fix}(\sigma)$. The $F_0$-algebra
$S_f=D[t;\sigma]/D[t;\sigma]f(t)$
 is called a \emph{(generalized) nonassociative cyclic algebra of degree $m$}.  We denote this algebra by $(D,\sigma, d)$ and call
$1,e,\dots,e^n,t,et\dots,e^{n-1}t,\dots e^nt^{m-1}$  its \emph{canonical basis} as a left $K$-vector space.
\end{definition}

\begin{remark} \label{re:nonassquats}
 If $c \in S \setminus S_0$, then $(S/S_0,\sigma,c)$ has nucleus
$S$  and center $S_0$. These algebras  first appeared over finite fields
 in \cite{S}, over general fields they were studied in \cite{S12}, and
 over number fields, in  \cite{SPO12}.
 If $c\in S_0^\times$, $S[t;\sigma]/S[t;\sigma] (t^m-c)$ is a classical associative cyclic algebra,
 cf. \cite{DO}, \cite{OS}.  If $c=0$, $S[t;\sigma]/S[t;\sigma] (t^m)$ is a commutative associative algebra,
 the direct product of $m$ copies of $S$. If $S/S_0$ is a cyclic
 Galois field extension of degree $m$ with Galois group $\langle \sigma\rangle$ and $c\in S\setminus S_0$,
 then ${\rm Nuc}((S/S_0,\sigma,c))=K$. If  $m$ is prime then $(S/S_0,\sigma,c)$ is  a division algebra. For non-prime $m$, a division
 algebra for all choices of $c$ such that $1,c,\dots,c^{m-1}$ are linearly independent \cite{S12}.
\end{remark}

\begin{example}\label{ex:gencyclic}
Let $F$ and $L$ be fields, $F_0=F\cap L$, and let $K$ be a cyclic field extension of both $F$ and $L$ such that
${\rm Gal}(K/F) = \langle \rho\rangle$ and $[K:F] = n$, ${\rm Gal}(K/L) = \langle \sigma \rangle$ and $[K:L] = m$,
such that $\rho$ and $\sigma$ commute.
 Let $D=(K/F, \rho, c)$  be an associative cyclic division algebra over $F$ of degree $n$ with canonical basis
 $1,e,\dots,e^{n-1}$ (where $e^n=c,$ $el=\rho(l)e$ for every $l$ in $K$), and  $c\in F_0$.
 For $x= x_0 + x_1 e+x_2  e^2+\dots + x_{n-1}e^{n-1}\in D$, extend $\sigma$ to an automorphism $\sigma\in {\rm Aut}_L(D)$
 of order $m$  via
$$\sigma(x)=\sigma(x_0) +  \sigma(x_1)e +\sigma(x_2)  e^2 +\dots +\sigma(x_{n-1}) e^{n-1}.$$
 For all $d \in D^\times$,
$S_f=D[t;\sigma]/D[t;\sigma](t^m-d)$
is the generalized nonassociative cyclic algebra $(D,\sigma,d)$ of dimension $m^2n^2$ over $F_0$.
For all $d \in F^\times$, we have
$$S_f=D[t;\sigma]/D[t;\sigma](t^m-d)= (L/F_0,\rho, c)\otimes_{F_0} (F/F_0, \sigma, d)=(D, \sigma, d)$$
$(D,\sigma,d)$  is associative if and only if $d\in F_0$.
For $f\in F_0[t]$, $(D,\sigma,d)$ is a \emph{generalized cyclic algebra}  of degree $n$ \cite[Section 1.4]{J96}.
\end{example}

\subsection{Space-time block coding}\label{subsec:STBC}

An ($s\times t$) \emph{space-time block code} (STBC) is a set $\mathcal{C}$ of complex $s\times t$ matrices.
$\mathcal{C}$ is called \emph{linear} if $X,X'\in \mathcal{C}$ implies $X\pm X'\in \mathcal{C}$.
 A linear code is called {\it fully diverse}, if ${\rm det}X\not=0$ for all $0\not=X\in \mathcal{C}$.

Let $K/F$ be a Galois field extension of degree $n$ and  $K$
an imaginary number field.
Nonassociative cyclic division algebras $A=(K/F,\sigma, c)$ of degree $n$
can be  used to build linear $n\times n$ STBCs with entries in $K$, since
 the right multiplication in $A$ induces the
 injective $K$-linear map $\gamma: A \hookrightarrow {\rm End}_K(A)\hookrightarrow {\rm Mat}_(K)$,
 $ a \mapsto R_a\mapsto M(a)$ (cf. Section \ref{sec:S_f}).
  The set of matrices $\gamma(A)$ is a  linear STBC that is fully diverse since $A$ is a division algebra.

Let  $A=(D,\sigma,d)$ be a generalized nonassociative cyclic division algebra, with
$D=(K/F,\sigma, c)$ an associative cyclic algebra of degree $n$. Again, $A$ can be  used to build a
fully diverse  linear $mn\times mn$ STBC
with entries in $K$: we know $\gamma: A \hookrightarrow {\rm End}_D(A), a\mapsto R_a$ is an injective $D$-linear
map, and $K\subset D$. Using the canonical $K$-basis of $A$, we  obtain an
$mn\times mn$-matrix $M(a)$ representing $R_a$ for every $a\in A$.
Thus we have $\gamma: A \hookrightarrow {\rm End}_D(A)\hookrightarrow {\rm Mat}_{mn}(K),$ $ a \mapsto R_a\mapsto M(a)$
and $\gamma(A)$ is a fully diverse  linear STBC.

When $d\in L^\times$ or $d\in F^\times$,
$\gamma(A)$ is used for the codes in \cite{P13.2}, \cite{PS15}, \cite{R13}.
For $m=2$, $\gamma(A)$ is used in the iterated codes constructed in \cite{MO13}.
 In particular, for $d\in F^\times$ the algebra in Example \ref{ex:gencyclic} is employed for the
 space-time block codes in \cite{R13}, see also \cite{SP14}.

\subsection{Cyclic $(f,\sigma,\delta)$-codes}

 Let $f\in S[t;\sigma,\delta]$ be monic of degree $m$ and $\sigma$ injective. We associate to an element
$a(t)=\sum_{i=0}^{m-1}a_it^i$ in $S_f$ the  vector $(a_0,\dots,a_{m-1})$. A
\emph{linear code of length $m$ over $S$} is a submodule of the $S$-module $S^m$.
 Conversely, for any linear code $\mathcal{C}$ of length $m$ we denote by $\mathcal{C}(t)$ the set of skew polynomials
 $a(t)=\sum_{i=0}^{m-1}a_it^i\in S_f$ associated to the codewords $(a_0,\dots,a_{m-1})\in \mathcal{C}$.

A \emph{cyclic $(f,\sigma,\delta)$-code} $\mathcal{C}\subset S^m$ is a set consisting of the vectors
$(a_0,\dots,a_{m-1})$ obtained from elements $h=\sum_{i=0}^{m-1}a_it^i$
in a left principal ideal $S_f g$ where $S_f=S[t;\sigma,\delta]g/S[t;\sigma,\delta]f$, and $g$ is a monic
right divisor of $f$.
A code $\mathcal{C}$ over $S$ is called
   $\sigma$-\emph{constacyclic} if there is a non-zero $c\in S$ such that
   $$(a_0,\dots,a_{m-1})\in  \mathcal{C}\Rightarrow (\sigma(a_{m-1})c,\sigma(a_0),\dots,\sigma(a_{m-2}))\in  \mathcal{C}.$$

\begin{lemma} (cf. \cite[Proposition 7]{P15})  \label{prop:skewcodemain}
Let $f\in R=S[t;\sigma,\delta]$ be monic of degree $m$.
 \\ (a) Let $\sigma$ be injective. Then:
 \begin{itemize}
 \item  Every  right divisor $g$ of $f$ of degree $<m$ with an invertible leading coefficient
generates a principal left ideal in $S_f$.
 \item All left ideals in $S_f$ which contain a non-zero  polynomial $g$
of minimal degree with invertible leading coefficient are principal left ideals, and $g$ is a right divisor of $f$ in $R$.
\item (\cite[Theorem  1]{BL13}) Each principal left ideal generated by a monic right divisor of $f$ is
an $S$-module which is isomorphic to a submodule of $S^m$ and forms a code of length $m$ and dimension $m-{\rm deg}(g)$.
 \end{itemize}
 (b) Let $S$ be a division ring.
Then all left ideals in $S_f$ are generated by some monic right divisor $g$ of $f$ in $R$.
\end{lemma}

\begin{proof}
(a)
 Let $g(t)$ be such a right divisor of $f(t)$, then the ideal $Rf$ is contained in $Rg$ and it is easy to check that
$Rg/Rf=\{h\in R_m\,|\,
h=sg \text{ for some } s\in R_m\}$ is a left ideal in $S_f$.
 \\  The proof of the second assertion is similar to the one of \cite[Lemma 1]{BGU07}: Suppose that $I$ is a left ideal in $S_f$ which contains a non-zero polynomial $g$
of minimal degree with invertible leading coefficient. For any $p\in I\subset R_m$, a right division by
$g$ yields unique $r,q\in R$ with
${\rm deg}(r)<{\rm deg}(g)$ such that
$p=qg+r $
and hence $r=p-qg\in I$. Since we chose $g\in I$ to have minimal degree, we conclude that $r=0$, implying
$p=qg$ and so $I=Rg$ is a principal left ideal, and $g$ is a right divisor of $f$ in $R$.
\\ (b)
Let $I$ be a left ideal of $S_f$. If $I=\{0\}$ then $I=(0)$. So suppose $I\not=(0)$ and choose a monic non-zero
polynomial $g$ in $I\subset R_m$
of minimal degree. As in the proof of (i), for any $p\in I$, a right division by $g$ yields unique
$r,q\in R$ with ${\rm deg}(r)<{\rm deg}(g)$ such that
$p=qg+r $ and hence $r=p-qg\in I$. Since $g\in I$ has minimal degree,  $r=0$, and so $I=Rg$.
\end{proof}

 Let $f,g,h,h'\in S[t;\sigma,\delta] $ be  monic polynomials such that $f=gh=h'g$.
Let $\mathcal{C}$ be the cyclic $(f,\sigma,\delta)$-code corresponding to $g$ and
$c(t)=\sum_{i=0}^{m-1}c_it^i\in S[t;\sigma,\delta]$.
Then  $(c_0,\dots,c_{m-1})\in \mathcal{C}$ is equivalent to
 $c(t)h(t)=0$ in $S_f$ \cite[Theorem 2]{BL13}, i.e. $h$ is a parity check polynomial for $\mathcal{C}$.

The codes $\mathcal{C}$ of length $m$ we consider consist of all elements $(a_0,\dots,a_{m-1})$ obtained  from polynomials
$a(t)=\sum_{i=0}^{m-1}a_it^i$ in a left principal ideal $S_f g$ of $S_f$, with $g$ a monic right divisor of $f$;
 $\sigma$-constacyclic codes are obtained when $f(t)=t^m-c\in S[t;\sigma]$.

For a field $K$, every skew polynomial ring $K[t;\sigma,\delta]$ can be made into either a twisted or a
differential polynomial ring by a linear change of variables  \cite[1.1.21]{J96}. When constructing linear codes, however,
we will consider general skew polynomial rings. They might
  produce better distance bounds than cyclic $(f,\sigma,\delta)$-codes constructed only with an automorphism,
   where $\delta=0$,  see \cite{BU14} for examples of this phenomenon.

%*******************************************************************************************%
%
% Natural orders I
%
%*******************************************************************************************%

\section{Natural orders in $S_f$ and their quotients by a prime ideal, I} \label{sec:naturalI}

In the following, we use the notation from  \cite[Section 2]{DO}.
Let $K/F$ be a Galois extension of number fields of degree $n$ with
$\mathcal{O}_F$ and $\mathcal{O}_K$  the rings of integers of $F$, respectively  $K$.

\subsection{} \label{subsec:assumptions}

Let $\mathfrak{p}$ be a maximal ideal of $\mathcal{O}_F$, $p$ the prime lying below $\mathfrak{p}$ and
$\mathcal{O}_F/\mathfrak{p}=\mathbb{F}_{p^j}$, where $j$ is the inertial degree of
$\mathfrak{p}$ above $p$. Let $\pi:\mathcal{O}_K\longrightarrow \mathcal{O}_K/\mathfrak{p}\mathcal{O}_K$ be
the canonical projection.
Let  $\sigma\in G={\rm Gal}(K/F)$.
We have $\sigma(\mathfrak{p}\mathcal{O}_K)\subset \mathfrak{p}\mathcal{O}_K$ since
$\sigma|_F=id$. Thus $\sigma$ induces a ring homomorphism
$$\overline{\sigma}:\mathcal{O}_K/\mathfrak{p}\mathcal{O}_K \longrightarrow \mathcal{O}_K/\mathfrak{p}\mathcal{O}_K,
\quad a+\mathfrak{p}\mathcal{O}_K \mapsto \sigma(a)+\mathfrak{p}\mathcal{O}_K $$
with $\pi\circ \sigma|_{\mathcal{O}_K}=\overline{\sigma} \circ \pi$ and ${\rm Fix}(\overline{\sigma})=\mathbb{F}_{p^j}$.
Suppose that $\delta$ is an $F$-linear left $\sigma$-derivation on $K$ such that
$\delta(\mathcal{O}_K)\subset \mathcal{O}_K$.
 Then $\delta$ induces a left $\overline{\sigma}$-derivation
 $\overline{\delta}:\mathcal{O}_K/\mathfrak{p}\mathcal{O}_K \longrightarrow \mathcal{O}_K/\mathfrak{p}\mathcal{O}_K.$
 Since $\mathcal{O}_K$ is a Dedekind domain we have
$$\mathfrak{p}\mathcal{O}_K=\mathfrak{p}_1^{e_1}\mathfrak{p}_2^{e_2}\dots \mathfrak{p}_g^{e_g}$$
for suitable prime (maximal) ideals $\mathfrak{p}_i$ of $\mathcal{O}_K$, $e_i\geq 0$. The ideals $\mathfrak{p}_t^{e_t}$
 are pair-wise comaximal.
By the Chinese Remainder Theorem, we have thus the following direct sum of rings:
$$\mathcal{O}_K/\mathfrak{p}\mathcal{O}_K= \mathcal{O}_K/ \mathfrak{p}_1^{e_1} \cdots\mathfrak{p}_g^{e_g}\mathcal{O}_K
\cong
\mathcal{O}_K/ \mathfrak{p}_1^{e_1}\mathcal{O}_K\times\dots\times \mathcal{O}_K/ \mathfrak{p}_g^{e_g}\mathcal{O}_K.$$
$G$ acts trivially on each of these $ \mathfrak{p}_t$, therefore there is an induced action of $G$ on each
$\mathcal{O}_K/ \mathfrak{p}_t^{e_t}\mathcal{O}_K$ and the above is an isomorphism of $G$-modules (cf. \cite[(9)]{OS}).
That means on each ring $\mathcal{O}_K/ \mathfrak{p}_t^{e_t}\mathcal{O}_K$ there is a canonical induced automorphism $\overline{\sigma}$
and a canonical left $\overline{\sigma}$-derivation $\overline{\delta}$ induced by $\delta$.

In particular, if $\mathfrak{p}$ is inert in $K/F$,
$\mathfrak{p}\mathcal{O}_K$ is a prime ideal
in $\mathcal{O}_K$ and thus $\mathcal{O}_K/\mathfrak{p}\mathcal{O}_K=\mathbb{F}_{p^{nj}}$ a finite field, and
$ \overline{\sigma}\in {\rm Gal}(\mathbb{F}_{p^{nj}}/\mathbb{F}_{p^{j}})$
(cf. \cite[Section 2]{DO} if $\delta=0$ and $K/F$ is cyclic).

%%%%%%%%%%%%%%%%%%%%%%%%%%%%%%%%%%%%%%%%%%%%%%%%%%%%%%%%%%%%%%%%%%%%%%%%%%%

\subsection{ }  \label{subsec:naturalI}

%%%%%%%%%%%%%%%%%%%%%%%%%%%%%%%%%%%%%%%%%%%%%%%%%%%%%%%%%%%%%%%%%%%%%%%%%%%%%%%%%

Suppose
$f(t)=\sum_{i=0}^{m}d_it^i\in \mathcal{O}_K[t;\sigma,\delta]$
 is monic and irreducible in $K[t;\sigma,\delta]$.
 Consider the nonassociative division algebra
$$S_f=K[t;\sigma,\delta]/K[t;\sigma,\delta] f$$
 over $F$.
Then the nonassociative  $\mathcal{O}_{F}$-algebra
$$\Lambda=\mathcal{O}_K[t;\sigma,\delta]/\mathcal{O}_K[t;\sigma,\delta] f$$
 is an $\mathcal{O}_{F}$-order in $S_f$ called the \emph{natural order} and
 $\Lambda=\mathcal{O}_K\oplus \mathcal{O}_K t\oplus \dots\oplus \mathcal{O}_K t^{m-1}$
as left $\mathcal{O}_K$-module.
Since $f$ is irreducible in $K[t;\sigma,\delta]$, $\Lambda$ does not have any zero divisors.
 The center of $\Lambda$ contains  $\mathcal{O}_{F}$.
Hence for every maximal ideal $\mathfrak{p}$ in $\mathcal{O}_{F}$,
$\mathfrak{p}\Lambda$ is a two-sided ideal of $\Lambda$.

$\Lambda$ is usually not maximal,
 but  it is uniquely determined whenever $Rf$ is not a two-sided ideal, since  in that case $K$ is the left and middle nucleus
 of $S_f$ and uniquely
 determines $\mathcal{O}_{K}$ and in turn $\Lambda$.
 (For examples of classes of maximal orders in nonassociative cyclic algebras of degree two, cf. \cite{Ka, LW}, the results
 there can be generalized to nonassociative algebras of any degree $n$.)

\begin{remark}
  $S_{f}$ is associative  if and only if $Rf$ is a two-sided ideal \cite[Theorem 4 (ii)]{P15}. Therefore our
   definition
 generalizes the non-commutative natural orders  in \cite{DO} which were only defined for two-sided ideals
 $Rf$ and $\delta=0$.
\end{remark}

 For any $g(t)=\sum_{i=0}^{m-1}a_i t^i\in \mathcal{O}_K[t;\sigma,\delta]$ define
   $\overline{g}(t)=\sum_{i=0}^{m-1}\overline{a_i}t^i\in
 (\mathcal{O}_K/\mathfrak{p}\mathcal{O}_K)[t;\overline{\sigma},\overline{\delta}]$ with $\overline{a_i}=
 a_i+\mathfrak{p}\mathcal{O}_K$.
  Let $\overline{f}(t)=\sum_{i=0}^{m}\overline{d_i}t^i\in (\mathcal{O}_K/\mathfrak{p}\mathcal{O}_K)[t;\overline{\sigma},\overline{\delta}]$
  with $\overline{d_i}= d_i+\mathfrak{p}\mathcal{O}_K$.

 \begin{lemma}
 (i)
 The surjective homomorphism of nonassociative rings
$$\Psi:\Lambda\longrightarrow  S_{\overline{f}},\quad g\mapsto \overline{g}$$
has kernel $\mathfrak{p}\Lambda$.
\\ (ii)  $\Psi$
induces an $\mathbb{F}_{p^j}$-algebra isomorphism  given by
$$\Psi:\Lambda/\mathfrak{p}\Lambda\longrightarrow
S_{\overline{f}},\quad g+\mathfrak{p}\Lambda \mapsto \overline{g}.$$
 \end{lemma}

 \begin{proof}
(i)
 $\Lambda$ is nonassociative  $\mathcal{O}_{F}$-algebra and $\Psi$ is a well-defined surjective homomorphism
with kernel $\mathfrak{p}\Lambda$: For all $g=\sum_{i=0}^{m-1}b_it^i\in \mathfrak{p}\Lambda$ it follows that $\overline{g}
=\sum_{i=0}^{m-1}\overline{b_i}t^i=0$  in
 $(\mathcal{O}_K/\mathfrak{p}\mathcal{O}_K)[t;\overline{\sigma},\overline{\delta}]$, so that
 $\mathfrak{p}\Lambda\subset ker(\Psi)$.

 Suppose conversely there is a nonzero $g\in\Lambda$ such that $\Psi(g)=\overline{g}=0$, then
$g=hf+r$ in $\mathcal{O}_K[t;\sigma,\delta]$ with a nonzero $r\in \mathcal{O}_K[t;\sigma,\delta]$, and so
$\overline{g}=\overline{h}\,\overline{f}+\overline{r}$ with $\overline{r}=0$ in
$(\mathcal{O}_K/\mathfrak{p}\mathcal{O}_K)[t;\overline{\sigma},\overline{\delta}]$.
 We have
$\mathfrak{p}\Lambda=\{al\,|\, a\in\mathfrak{p},l\in\Lambda\}=\{\sum_{i=0}^{m-1}a_it^i\,|\, a_i\in
 \mathfrak{p}\mathcal{O}_K\}.$
This implies that $r\in \mathfrak{p}\mathcal{O}_K[t;\sigma,\delta]=\mathfrak{p}\Lambda$.
\\ (ii) follows from (i).
 \end{proof}

\begin{example} \label{ex:nca1}
Let  ${\rm Gal}(K/F)=$ $<\sigma>$ and $f(t)=t^n-d\in \mathcal{O}_K[t;\sigma]$  irreducible in $K[t;\sigma]$.
 $A=(K/F,\sigma,d)$ is a  nonassociative cyclic division algebra of degree $n$ over $F$ and
 $S_{\overline{f}}=\big((\mathcal{O}_K/\mathfrak{p}\mathcal{O}_K)/\mathbb{F}_{p^j},\overline{\sigma},\overline{d}\big)$
with $\overline{f}(t)=t^n-\overline{d}\in (\mathcal{O}_K/\mathfrak{p}\mathcal{O}_K)[t;\overline{\sigma}].$
If $d\in \mathcal{O}_F$ is non-zero, $A$ is associative
and $\Lambda$ depends on the choice of the maximal subfield $K$ in $A$. Then $S_{\overline{f}}$
is an associative (generalized) cyclic algebra as in \cite{DO, OS} and $\overline{f}(t)$ is reducible
whenever $\mathcal{O}_K/\mathfrak{p}\mathcal{O}_K$ is a field.

If $d\in \mathcal{O}_K\setminus \mathcal{O}_F$, $A$ is not associative and the natural order
$$\Lambda=\mathcal{O}_K[t;\sigma]/\mathcal{O}_K[t;\sigma]f=\mathcal{O}_K \oplus  \mathcal{O}_K t\oplus \dots \oplus\mathcal{O}_K  t^{n-1}$$
of $A$ is uniquely determined.
If $n$ is prime then $f$ is irreducible and $A$  a division algebra for every $d\in \mathcal{O}_K\setminus \mathcal{O}_F$.
 If $n$ is not prime and $1,d,\dots,d^{n-1}$ are linearly independent then $f$ is irreducible and
$A$ a division algebra (Remark \ref{re:nonassquats}). Furthermore,
$$\Lambda/\mathfrak{p}\Lambda\cong ((\mathcal{O}_K/\mathfrak{p}\mathcal{O}_K)/\mathbb{F}_{p^j},
\overline{\sigma},\overline{d})=S_{\overline{f}}.$$
\end{example}

\begin{remark}\label{re:I}
Suppose that $K/F$ is cyclic of degree $n$ and inertial with respect to $\mathfrak{p}$, then $\mathcal{O}_K/\mathfrak{p}\mathcal{O}_K=
\mathbb{F}_{p^{jn}}$
 and  ${\rm Gal}(\mathbb{F}_{p^{jn}}/\mathbb{F}_{p^{j}})=\langle\overline{\sigma}\rangle$.
 Let $f(t)=\sum_{i=0}^{n}d_it^i\in \mathcal{O}_K[t;\sigma,\delta]$ be such that ${\overline{f}}(t)=t^n-\overline{ d_0}$.
 \\ (i)
In \cite{DO}, only polynomials $f(t)=t^n-d_0$ with $d_0\in \mathcal{O}_F$
 are considered which makes the ideal
 $(\mathcal{O}_K/\mathfrak{p}\mathcal{O}_K)[t;\overline{\sigma}]\overline{f}$
 two-sided and the resulting $\mathbb{F}_{p^j}$-algebra associative. In this case,
$\overline{f}(t)=t^n-\overline{d_0}$ is always reducible in $\mathbb{F}_{p^{jn}}[t;\overline{\sigma}]$.
\\ (ii) By Lemma \ref{prop:skewcodemain},
if $\overline{f}$ is irreducible, then $S_{\overline{f}}$ has no non-trivial left ideals.
 For instance, if $n$ is prime and $d_0\not\in \mathcal{O}_F$  then
 for all $\overline{d_0}\not=0$,
 $$\Lambda/\mathfrak{p}\Lambda\cong (\mathbb{F}_{p^{jn}}/\mathbb{F}_{p^j},\overline{\sigma},\overline{d_0})$$
 is always  a division algebra, i.e. $f(t)=t^n-\overline{d_0}$ is irreducible, and so there are no non-trivial left ideals
 by Lemma \ref{prop:skewcodemain} (b).
\end{remark}

%*******************************************************************************************%
%
% Lattice encoding I
%
%*******************************************************************************************%

\section{Lattice encoding of cyclic $(f,\sigma,\delta)$-codes over $\mathcal{O}_K/\mathfrak{p}\mathcal{O}_K$, I}
\label{sec:codes}

We keep the assumptions and notation from Section \ref{sec:naturalI}.
Let  $\mathcal{I}=\Lambda g(t)$ be a principal left ideal of $\Lambda$   generated by a monic polynomial $g(t)$ such that
$\mathfrak{p}\subset \mathcal{I}\cap\mathcal{O}_F$.
Then $\mathcal{I}/\mathfrak{p}\Lambda$ is a principal left ideal of $\Lambda/\mathfrak{p}\Lambda$ and
$\Psi(\mathcal{I}/\mathfrak{p}\Lambda)$
is a principal left ideal of $((\mathcal{O}_K/\mathfrak{p}\mathcal{O}_K)/\mathbb{F}_{p^j},\overline{\sigma},\overline{c})$
generated by the monic polynomial $\Psi(g+\mathfrak{p}\Lambda)=\overline{ g}$.
That means, $\Psi(\mathcal{I}/\mathfrak{p}\Lambda)$ corresponds to an
$(\overline{f},\overline{\sigma},\overline{\delta})$-code $\mathcal{C}$ over $\mathbb{F}_q$.
In particular, if we choose $f(t)$ such that $\overline{f}(t)=t^m-\overline{ c}$ with $\overline{ c}$ non-zero,
 then $\Psi(\mathcal{I}/\mathfrak{p}\Lambda)$ corresponds to a $\overline{\sigma}$-constacyclic
code over $\mathbb{F}_q$.

If $\overline{f}$ is irreducible and $\mathcal{O}_K/\mathfrak{p}\mathcal{O}_K$ a field,
then $S_{\overline{f}}$   has no nontrivial principal left ideals which contain a non-zero  polynomial
of minimal degree with invertible leading coefficient and so  $\mathcal{C}$ has length $n$ and dimension $n$,
or is zero, whereas when $\overline{f}$ is reducible and $\mathcal{O}_K/\mathfrak{p}\mathcal{O}_K$ a field, an
$(\overline{f},\overline{\sigma},\overline{\delta})$-code $\mathcal{C}$ corresponds to a right divisor $\overline{g}$ of
$\overline{f}$ and has dimension  $n-{\rm deg}(\overline{g})$. So we will look for irreducible $f$ where
$\overline{f}$ is reducible.

\subsection{Construction A} \label{ex:nca}

Let
$$\rho: \Lambda\longrightarrow \Lambda/\mathfrak{p}\Lambda \longrightarrow \Psi(\Lambda/\mathfrak{p}\Lambda)$$ be the
canonical projection $\Lambda\longrightarrow \Lambda/\mathfrak{p}\Lambda$ composed with $\Psi$.
We know that $\mathcal{O}_K$ is a free $\mathbb{Z}$-module of rank $n[F:\mathbb{Q}]$.
Then
$$L=\rho^{-1}(\mathcal{C})=\mathcal{I}$$
is a $\mathbb{Z}$-module of dimension $N=nm[F:\mathbb{Q}]$.
The embedding of this lattice into $\mathbb{R}^N$ is canonically determined by
 considering $A\otimes_{\mathbb{Q}}\mathbb{R}$. Now all works exactly as
 as explained in \cite[Section 3.3]{DO}.
The construction of $L$ can  be seen as a non-commutative variation of the classical Construction A in \cite{CS}.

This way we can construct a lattice $L$ in $\mathbb{R}^N$ from the linear code $\mathcal{C}$ over the finite ring
$S=\mathcal{O}_K/\mathfrak{p}\mathcal{O}_K$.
The non-commutative variation of Construction A in \cite{DO} is the special
case that $f(t)=t^n-c\in \mathcal{O}_F[t]\subset K[t;\sigma]$, where $S_f$ is associative.

\begin{example}
Let $K[t;\sigma]=\mathbb{Q}(i)[t;\sigma]$ with $\sigma$ the complex conjugation, so that
$F=\mathbb{Q}$, $\mathcal{O}_F=\mathbb{Z}$ and $\mathcal{O}_K=\mathbb{Z}[i]$.
Let $f(t)=t^2-t+(i-3)\in\mathbb{Z}[i][t,\sigma]$, then $f(t) $ is irreducible in $\mathbb{Q}(i)[t;\sigma]$, since
  $\sigma(z)z-z\not=i-3$ for all $z\in \mathbb{Q}(i)$ \cite[(17)]{P66}.
  Let $p=3$. Then $\mathbb{Z}[i]/3\mathbb{Z}[i]=\mathbb{F}_{9}$ and using
 the natural order $\Lambda=\mathbb{Z}[i]\oplus \mathbb{Z}[i]t$ in $S_f$, we obtain the nonassociative algebra
 $\Lambda/3\Lambda\cong S_{\overline{f}}$  over $\mathbb{F}_{3}$ with
  $\overline{f}(t)=t^2-(\alpha+\alpha^3-1)t+\alpha-3\in \mathbb{F}_{9}[t;\overline{\sigma}]$
 where
 $\overline{\sigma}(\alpha)=\alpha^3$, if $\alpha$ is a primitive root of $\mathbb{F}_{9}$ over $\mathbb{F}_{3}$, that is
$\alpha^2+1=0$. Since
$$\overline{f}(t)=(t-2+\alpha)(t+1+\alpha),$$
$\overline{f}$ is reducible in $\mathbb{F}_{9}[t;\overline{\sigma}]$. The left ideal generated by
$(t+1+\alpha)$ in $S_{\overline{f}}$ yields a cyclic
 $(f,\sigma,0)$-code of length 2 and dimension one. Taking the pre-image of it under $\Psi$ it corresponds to a
 principal left ideal $\mathcal{I}/3\Lambda$ in $\Lambda/3\Lambda$.
\end{example}

\subsection{Examples involving nonassociative quaternion algebras}

Let $K=\mathbb{Q}(i)$, $F=\mathbb{Q}$, so that $\mathcal{O}_F=\mathbb{Z}$ and $\mathcal{O}_K=\mathbb{Z}[i]$.
The examples given in \cite{DO} are special cases of our construction using
 cyclic algebras. We now consider some algebras which are not associative.

 Let $f(t)=t^2-bt-c\in\mathbb{Z}[i][t,\sigma]$ be irreducible in $\mathbb{Q}(i)[t;\sigma]$.
This is equivalent to $\sigma(z)z-bz-c\not=0$ for all $z\in \mathbb{Q}(i)$ \cite[(17)]{P66}. In particular,
 if $b,c\in \mathbb{Z}$ then $f(t)$ is irreducible if $b^2+4c<0$ (alternatively, if $f$ is an irreducible
 polynomial in $\mathbb{R}$) by \cite[Corollary 2.6]{BZ}.
 Suppose that $\overline{f}(t)=t^2-\overline{c}\in (\mathbb{Z}[i]/\mathfrak{p}\mathbb{Z}[i])[t;\overline{\sigma}]$
 for some maximal ideal $\mathfrak{p}$ in $\mathcal{O}_F$.

 For the natural order $\Lambda=\mathbb{Z}[i]\oplus \mathbb{Z}[i]t$, we obtain the (perhaps nonassociative) quaternion algebra
 $$\Lambda/\mathfrak{p}\Lambda\cong ((\mathbb{Z}[i]/\mathfrak{p}\mathbb{Z}[i])/\mathbb{F}_{p^j},
 \overline{\sigma},\overline{c})=S_{\overline{f}}.$$
In particular,
$\Lambda/\mathfrak{p}\Lambda=(\mathbb{Z}[i]/\mathfrak{p}\mathbb{Z}[i])\oplus (\mathcal{O}_K/\mathfrak{p}\mathcal{O}_K)t$
as $(\mathbb{Z}[i]/\mathfrak{p}\mathbb{Z}[i])$-module.

For any choice of $c\in \mathbb{Z}$ such that $c\not\in\mathfrak{p}\mathbb{Z}[i]$, $\overline{f}(t)=t^2-\overline{c}\in\mathcal{O}_K[t,\sigma]$ is reducible.

For $b=0$ and any choice of $c\in \mathbb{Z}[i]\setminus\mathbb{Z}$, $f(t)=t^2-c\in \mathbb{Z}[i][t,\sigma]$ is irreducible
 in $\mathbb{Q}(i)[t;\sigma]$ and therefore
$$A=S_f=(\mathbb{Q}(i)/\mathbb{Q},\sigma,c)$$
 a nonassociative quaternion division algebra (for instance,  $ct=(t^2)t\not=t(t^2)=\sigma(c)t$ in $A$.) We can also write $A$ as the Cayley-Dickson doubling
${\rm Cay}(\mathbb{Q}(i),c)$, defined in the obvious way.

\begin{example}  \label{ex:1}
Let $f(t)=t^2-c\in \mathbb{Z}[i][t;\sigma]$, $c\in \mathbb{Z}[i]\setminus\mathbb{Z}$.
 Choose any $p$ which remains inert in $\mathbb{Q}(i)$, then
$\mathcal{O}_K/\mathfrak{p}\mathcal{O}_K=\mathbb{F}_{p^{2j}}$,  where $j$ is the inertial degree of
$\mathfrak{p}$ above $p$, and
$$\Lambda/\mathfrak{p}\Lambda\cong (\mathbb{F}_{p^{2j}}/\mathbb{F}_{p^j},\overline{\sigma},\overline{c}).$$
If $\overline{c}\not=0$ this is a division algebra because $f(t)=t^2-\overline{c}$ is irreducible. Given
any principal left ideal $\mathcal{I}$ of $\Lambda$
 containing $p$, $\Psi(\mathcal{I}/\mathfrak{p}\Lambda)$ is thus either trivial or all of
$S_{\overline{f}}=\mathbb{F}_{p^{2j}}[t;\overline{\sigma}]/\mathbb{F}_{p^{2j}}[t;\overline{\sigma}]\overline{f}$.

If $\overline{c}=0$, i.e. when $c\in \mathfrak{p}\mathcal{O}_K$
then $\overline{f}(t)=t^2$  and $S_{\overline{f}}$ is a commutative associative algebra.
 There are no
$\overline{\sigma}$-constacyclic codes since here $\overline{c}=0$. Thus this algebra
cannot be used for lattice encoding of $\overline{\sigma}$-constacyclic codes.

E.g., take $p=3$. If  $c=i$ then
$\Lambda/3\Lambda\cong (\mathbb{F}_{9}/\mathbb{F}_{3},\overline{\sigma},\overline{i})$
is a nonassociative quaternion division algebra over $\mathbb{F}_{3}$ where
 $\overline{\sigma}(\alpha)=\alpha^3$, if $\alpha$ is a primitive root of $\mathbb{F}_{9}$ over $\mathbb{F}_{3}$, that is
$\alpha^2+1=0$.
 $\mathcal{I}=(1+i)\Lambda$ satisfies $3\in \mathcal{I}\cap\mathcal{O}_F$
 (since $1-2i\in \mathbb{Z}[i]$, so $(1+i)(1-2i)=3\in \mathcal{I}$). Hence
$\mathcal{I}/3\Lambda$ is a  left principal ideal of
$\Lambda/3\Lambda\cong (\mathbb{F}_{9}/\mathbb{F}_{3},\overline{\sigma},\overline{i})$, generated by $\Psi((1+i)+3\Lambda)$,
 implying
$\mathcal{I}/3\Lambda \cong (\mathbb{F}_{9}/\mathbb{F}_{3},\overline{\sigma},\overline{i})$. Since here
$\overline{f}$ is irreducible, the only available (and trivial) $\overline{\sigma}$-constacyclic code here is the one
corresponding to the algebra $(\mathbb{F}_{9}/\mathbb{F}_{3},\overline{\sigma},\overline{i})$.

If $c=3i$, then $\overline{f}(t)=t^2$ and
$\Lambda/3\Lambda\cong (\mathbb{F}_{9}/\mathbb{F}_{3},\overline{\sigma},0)$
is a commutative associative algebra over $\mathbb{F}_{3}$. There are no
$\overline{\sigma}$-constacyclic codes since  $\overline{c}=0$, hence this example
cannot be used for lattice encoding of $\overline{\sigma}$-constacyclic codes.
\end{example}

\begin{example} \label{ex:2}
Let $f(t)=t^2-c\in \mathbb{Z}[i][t;\sigma]$, $c\in \mathbb{Z}[i]\setminus\mathbb{Z}$.
\\ (i)  Choose any $p$ which splits in $\mathbb{Q}(i)$, e.g. $p=5$. Then $(5)=(1+2i)(1-2i)$ means that
$\mathbb{Z}[i]/5 \mathbb{Z}[i]\cong \mathbb{Z}[i]/(1-2i) \times \mathbb{Z}[i]/(1+2i)\cong \mathbb{F}_5\times\mathbb{F}_5$
and
 $$\Lambda/5\Lambda=(\mathbb{F}_5\times\mathbb{F}_5)\oplus (\mathbb{F}_5\times\mathbb{F}_5) e$$
is a  nonassociative quaternion algebra over $\mathbb{F}_{5}$ with
$$\Lambda/5\Lambda\cong ((\mathbb{F}_{5}\times \mathbb{F}_{5})/\mathbb{F}_{5},\overline{\sigma},\overline{c})
=S_{\overline{f}} \text{ with }
\overline{f}(t)=t^2-\overline{c}\in (\mathbb{F}_{5}\times \mathbb{F}_{5})[t;\overline{\sigma}].$$
Here, $\overline{\sigma}(a,b)=(b,a)$ fixes the elements $(a,a)$, $a\in\mathbb{F}_{5}$.
The algebra $\Lambda/5\Lambda$ is a split nonassociative quaternion algebra \cite{W}, however for all
$\overline{c}\not=0$,
 $\overline{f}$ is irreducible in $(\mathbb{F}_{5}\times \mathbb{F}_{5})[t;\overline{\sigma}]$,
since $\overline{c}\not\in \mathbb{F}_{5} $. In this case, there are no non-trivial divisors of $\overline{f}$ and hence no non-trivial
codes to lift. If $\overline{ c}=0$ then $\overline{f}(t)=t^2$  and
$\Lambda/5\Lambda$ a commutative associative algebra.   There are no
$\overline{\sigma}$-constacyclic codes since  $\overline{c}=0$.
\\ (ii) Choose  $p=2$ which ramifies in $\mathbb{Q}(i)$. Then
$\mathbb{Z}[i]/2 \mathbb{Z}[i]\cong \mathbb{F}_2+\mathbb{F}_2 v=\{0,1,v,v+1\}$ with $v^2=0$. I.e.,
$\mathbb{F}_2+\mathbb{F}_2 v$ is the finite chain ring of characteristic $2$, nilpotency index $2$ and
residue field $\mathbb{F}_{2}$. Here $\overline{\sigma}=id$,
 $$\Lambda/2\Lambda=(\mathbb{F}_2+\mathbb{F}_2 v)\oplus (\mathbb{F}_2+\mathbb{F}_2 v) e$$
 and we have the following $\mathbb{F}_{2}$-algebra isomorphism:
$$\Lambda/2\Lambda\cong ((\mathbb{F}_2+\mathbb{F}_2 v)/\mathbb{F}_{2},id,\overline{c})=S_{\overline{f}}$$
with $\overline{f}(t)=t^2-\overline{c}\in (\mathbb{F}_2+\mathbb{F}_2 v)[t]$,
  $c\in \mathbb{Z}[i]\setminus \mathbb{Z}$.
For both  $\overline{c}=v$ and $\overline{c}=v+1$,
it is easy to show that
$\overline{f}$ is irreducible, and if $\overline{c}=0$ again $\overline{f}=t^2.$
We conclude that $p=2$ does not yield an algebra which can be employed for lattice encoding.
\end{example}

\subsection{}
 For
 a nonassocative cyclic algebra $A=(K/F,\sigma,c)$ of prime degree $n$, $A$ is a division algebra if and only if
$c\in K\setminus F$. Examples \ref{ex:1} and \ref{ex:2} demonstrate that this poses a problem when trying to find
irreducible $f(t)=t^n-c$ such that $\overline{f}(t)=t^n-\overline{ c}$  is reducible and $0\not=\overline{ c}$,
 since  $\overline{f}(t)$ is either
irreducible, or $\overline{ c}=0$.  This is not the case when $n$ is not prime:

\begin{example} \label{ex:4}
Let $f(t)=t^4-c$.
Let $\omega_{15}$ be a primitive 15th root of unity, $K=\mathbb{Q}(i,\omega_{15}+\omega_{15}^{-1})$, and
$F=\mathbb{Q}(i)$. Choose $c\in \mathcal{O}_K\setminus \mathcal{O}_F$
such that $1,c,c^2,c^3$ are $F$-linearly independent. Then
$D=(\mathbb{Q}(i,\omega_{15}+\omega_{15}^{-1})/\mathbb{Q}(i),\sigma,c)$ is a nonassociative cyclic division algebra
of degree 4 and
$$\Lambda=\mathbb{Z}[i,\omega_{15}+\omega_{15}^{-1}]\oplus \mathbb{Z}[i,\omega_{15}+\omega_{15}^{-1}] t\oplus
\mathbb{Z}[i,\omega_{15}+\omega_{15}^{-1}]t^2\oplus \mathbb{Z}[i,\omega_{15}+\omega_{15}^{-1}]t^3$$
is the natural order in $D$.

Let $\mathfrak{p}=( 1+i)$. Then $\mathfrak{p}$ is unramified in
$\mathbb{Q}(i,\omega_{15}+\omega_{15}^{-1})$,  $\mathbb{Z}[i]/( 1+i)\cong\mathbb{F}_2$, and
$$\Lambda/\mathfrak{p}\Lambda\cong(\mathbb{F}_{16}/\mathbb{F}_2,\overline{\sigma},\overline{c})$$
 is a nonassociative cyclic
 algebra of degree $4$ over $\mathbb{F}_2$ which for $\overline{ c}\not=0$ is never a division algebra, since $1,c,c^2,c^3$ are always
 linearly dependent over $\mathbb{F}_2$. Hence $f(t)=t^4-\overline{c}$ is reducible.
 If $\overline{c}=1$ then given any principal left ideal $\mathcal{I}$ of $\Lambda$
 containing $(1+i)$ that is generated by a monic polynomial, $\Psi(\mathcal{I}/(1+i)\Lambda)$ corresponds to a
$\overline{\sigma}$-constacyclic code  over $\mathbb{F}_{16}$.
\end{example}

%*******************************************************************************************%
%
% Natural orders II
%
%*******************************************************************************************%

\section{Natural orders in $S_f$ and their quotients by a prime ideal, II} \label{sec:naturalII}

\subsection{}\label{subsec:iteratedI}

Let $K/F$ be a cyclic Galois extension of number fields of degree $n$ and let
 $D=(K/F, \rho, c)$ be a cyclic division algebra over $F$ such that $c\in \mathcal{O}_F^\times$.
 Let $\mathcal{D}=(\mathcal{O}_K/\mathcal{O}_F,\rho, c)$
be the generalized associative cyclic algebra over $\mathcal{O}_F$ of degree $n$ such that
 $\mathcal{D} \otimes_{\mathcal{O}_F}F=(K/F, \rho, c)=D$.
 Then
$\mathcal{D}=\mathcal{O}_K \oplus   \mathcal{O}_Ke \oplus \dots  \oplus\mathcal{O}_K  e^{n-1}$
is a natural $\mathcal{O}_F$-order of $D$, cf. \ref{subsec:naturalI} or \cite{DO}.

Let $\sigma\in {\rm Aut}(D)$ and $\delta$ be a $\sigma$-derivation on $D$, satisfying the following
criteria:
\begin{itemize}
\item $F_0=F\cap {\rm Fix}(\sigma)\cap {\rm Const}(\delta)$ is a number field.
\item $\sigma(\mathcal{D})\subset \mathcal{D}$ and $\delta(\mathcal{D})\subset \mathcal{D}$.
\item $S_0=\mathcal{O}_F\cap {\rm Fix}(\sigma)\cap {\rm Const}(\delta)$ is  the ring of integers of $F_0$ where here
$\sigma$ and $\delta$ denote the restrictions of $\sigma$ and $\delta$ to $\mathcal{D}$.
\end{itemize}

Suppose
$f(t)=\sum_{i=0}^{m}d_it^i\in \mathcal{D}[t;\sigma,\delta]$
 is monic and irreducible in
$D[t;\sigma,\delta]$. Consider the division algebra
$$S_f=D[t;\sigma,\delta]/D[t;\sigma,\delta] f$$
 over $F_0$. Then the $S_0$-algebra
$$\Lambda=\mathcal{D}[t;\sigma,\delta]/\mathcal{D}[t;\sigma,\delta] f$$
 is an $S_0$-order in $S_f$ which we call the \emph{natural order} (this is again usually not maximal). $\Lambda$ is not uniquely determined even when $Rf$ is not a two-sided ideal. It
depends on the choice of the maximal subfield in $D$ which we will assume to be $K$.

 Since $f$ is irreducible in $D[t;\sigma,\delta]$, $\Lambda$ does not have  zero divisors.
If $1,e,\dots,e^{n-1}$
is the canonical basis of $D$ then
$$\Lambda=\mathcal{O}_K \oplus  \mathcal{O}_K e\oplus \dots  \oplus \mathcal{O}_K e^{n-1}
\oplus \mathcal{O}_K t\oplus \mathcal{O}_K et\oplus\dots
\oplus \mathcal{O}_K e^{n-1}t \oplus \dots\oplus\mathcal{O}_K  e^{n-1}t^{m-1}$$
as left $\mathcal{O}_K$-module.

Let $\mathfrak{p}$ be a prime ideal in $S_0$ such that $\mathfrak{p}\mathcal{O}_F$ is maximal.  Since $S_0$
lies in the centers of both $\mathcal{D}$ and  $\Lambda$, $ \mathfrak{p}\mathcal{D}$ is a two-sided ideal of $\mathcal{D}$
and $\mathfrak{p}\Lambda$ is a two-sided ideal of $\Lambda$.
Let $\pi:\mathcal{D}\longrightarrow \mathcal{D}/\mathfrak{p}\mathcal{D}$ be
the canonical projection.
We have $\sigma( \mathfrak{p}\mathcal{D})\subset \mathfrak{p}\mathcal{D}$ since
 $\mathfrak{p}\subset{\rm Fix}(\sigma)$ and $\sigma(\mathcal{D})\subset \mathcal{D}$ by  assumption.
 Thus $\sigma$ induces a ring homomorphism
$$\overline{\sigma}:\mathcal{D}/\mathfrak{p}\mathcal{D} \longrightarrow \mathcal{D}/\mathfrak{p}\mathcal{D}$$
with
${\rm Fix}(\overline{\sigma})={\rm Fix}(\sigma)/\mathfrak{p}{\rm Fix}(\sigma)$
and $\pi\circ \sigma=\overline{\sigma} \circ \pi$.
We also have
$\delta( \mathfrak{p}\mathcal{D})\subset  \mathfrak{p}\mathcal{D}$ by  assumption, so that $\delta$ induces a left $\overline{\sigma}$-derivation
 $\overline{\delta}:\mathcal{D}/\mathfrak{p}\mathcal{D} \longrightarrow \mathcal{D}/\mathfrak{p}\mathcal{D}$
 with field of constants
 ${\rm Const}(\overline{\delta})={\rm Const}(\delta)/\mathfrak{p}.$
Let $$\overline{S_0}={\rm Fix}(\overline{\sigma})\cap {\rm Const}(\overline{\delta})\cap\overline{F} $$
with $\overline{F}=\mathcal{O}_F/\mathfrak{p}\mathcal{O}_F=\mathbb{F}_{p^j}$, where $j$ is the inertial degree of
$\mathfrak{p}$ above $p$.
 For any $g(t)=\sum_{i=0}^{m-1}a_i t^i\in \mathcal{D}[t;\sigma,\delta]$ define
   $\overline{g}(t)=\sum_{i=0}^{m-1}\overline{a_i}t^i\in
 (\mathcal{D}/\mathfrak{p}\mathcal{D})[t;\overline{\sigma},\overline{\delta}]$ with $\overline{a_i}=
 a_i+\mathfrak{p}\mathcal{D}$.  Let $\overline{f}(t)=\sum_{i=0}^{m}\overline{d_i}t^i\in (\mathcal{D}/\mathfrak{p}\mathcal{D})[t;\overline{\sigma},\overline{\delta}]$
  with $\overline{d_i}= d_i+\mathfrak{p}\mathcal{D}$, then
  $$S_{\overline{f}}=(\mathcal{D}/\mathfrak{p}\mathcal{D})[t;\overline{\sigma},\overline{\delta}]/(\mathcal{D}/\mathfrak{p}\mathcal{D})
[t;\overline{\sigma},\overline{\delta}]\overline{f}.$$
Since
$\overline{S_0}=S_0/\mathfrak{p} \cong {\rm Fix}(\overline{\sigma})\cap {\rm Const}(\overline{\delta})\cap\overline{F}$,
 $S_{\overline{f}}$ is an algebra over a subfield of $\mathbb{F}_{p^j}$.

 \begin{lemma}
 (i)
  The surjective homomorphism of additive groups
$$\Psi:\Lambda\longrightarrow  (\mathcal{D}/\mathfrak{p}\mathcal{D})[t;\overline{\sigma},\overline{\delta}]/
(\mathcal{D}/\mathfrak{p}\mathcal{D})[t;\overline{\sigma},\overline{\delta}]\overline{f},\quad g\mapsto \overline{g}$$
has kernel $\mathfrak{p}\Lambda$.
\\ (ii) $\Psi$ induces an $\overline{S_0}$-algebra isomorphism
$$\Psi:\Lambda/\mathfrak{p}\Lambda\cong
(\mathcal{D}/\mathfrak{p}\mathcal{D})[t;\overline{\sigma},\overline{\delta}]/(\mathcal{D}/\mathfrak{p}\mathcal{D})
[t;\overline{\sigma},\overline{\delta}]\overline{f},\quad g+\mathfrak{p}\Lambda \mapsto \overline{g}$$
 \end{lemma}

 \begin{proof}
 (i) For all $g=\sum_{i=0}^{m-1}b_it^i\in \mathfrak{p}\Lambda$ we have $\overline{g}=\sum_{i=0}^{m-1}\overline{b_i}t^i=0$ in
 $(\mathcal{D}/\mathfrak{p}\mathcal{D})[t;\overline{\sigma},\overline{\delta}]$, so that
 $\mathfrak{p}\Lambda\subset ker(\Psi)$.

 Suppose conversely there is a nonzero $g\in\Lambda$ such that $\Psi(g)=\overline{g}=0$, then
$g=hf+r$ in $\mathcal{D}[t;\sigma,\delta]$ with a nonzero $r\in \mathcal{D}[t;\sigma,\delta]$, and so
$\overline{g}=\overline{h}\,\overline{f}+\overline{r}$ with $\overline{r}=0$ in
$(\mathcal{D}/\mathfrak{p}\mathcal{D})[t;\overline{\sigma},\overline{\delta}]$.
 We have
$\mathfrak{p}\Lambda=\{al\,|\, a\in\mathfrak{p},l\in\Lambda\}=$
$\{\sum_{i=0}^{m-1}a_it^i\,|\, a_i\in \mathfrak{p}\mathcal{D}\}.$
This implies that $r\in \mathfrak{p}\mathcal{D}[t;\sigma,\delta]=\mathfrak{p}\Lambda$.
\\ (ii) follows from (i).
 \end{proof}

\subsection{A special case} \label{subsec:iterated}
Let $S_f=(D,\sigma,d)$ be the $F_0$-algebra constructed in Example  \ref{ex:gencyclic} where now
$F$, $L$ and $K$ be number fields. Suppose that $c \in \mathcal{O}_{F_0}$ and that $d\in \mathcal{O}_L^\times$ or
$d\in \mathcal{O}_F^\times$.
Then $\mathcal{D}=(\mathcal{O}_K/\mathcal{O}_F, \rho, c)$
is an associative cyclic algebra over $\mathcal{O}_F$ of degree $n$ such that
 $\mathcal{D} \otimes_{\mathcal{O}_F}F=(K/F, \rho, c)=D$ is a division algebra over $F$.
 For $x= x_0 + x_1 e+x_2 e^2 +\dots + x_{n-1}e^{n-1}\in D$, define
$$\sigma(x)=\sigma(x_0) +  \sigma(x_1)e + \sigma(x_2)e^2 +\dots + \sigma(x_{n-1})e^{n-1}.$$
Since $c \in \mathcal{O}_{F_0}$, $\sigma\in {\rm Aut}_L(D)$ has order $m$
and restricts to $\sigma\in {\rm Aut}_{\mathcal{O}_{L}}(\mathcal{D})$.

 Let $\mathfrak{p}$ be a prime ideal in $\mathcal{O}_{F_0}$ such that $\mathfrak{p}\mathcal{O}_F$ is maximal. Then
$$\Lambda/\mathfrak{p}\Lambda\cong
(\mathcal{D}/\mathfrak{p}\mathcal{D})[t;\overline{\sigma}]/(\mathcal{D}/\mathfrak{p}\mathcal{D})
[t;\overline{\sigma}]\overline{f}\cong (\overline{D}, \overline{\sigma}, \overline{d})$$
is an algebra over
 $\overline{F_0}=\mathcal{O}_{F_0}/\mathfrak{p}$,
with
$\overline{D}=\mathcal{D}/\mathfrak{p}\mathcal{D}$
 a generalized associative cyclic algebra over $\mathbb{F}_{p^j}={\rm Fix}(\overline{\rho})$.

\begin{example} \label{ex:5.2}
Let $\omega$ denote the primitive third root of unity, $\omega_7$  a primitive $7$th root of
unity  and $\theta = \omega_7 + \omega_7^{-1} = 2
\cos(\frac{2 \pi}{7})$. Put $F = \mathbb{Q}(\theta)$.
Let $K = F(\omega) =\mathbb{Q}(\omega, \theta)$ and consider
 the quaternion division algebra $D = (K/F, \sigma, -1)$. Note that
 $\sigma(\omega) = \omega^2$. Let $L =\mathbb{Q}(\omega)$, so that $K/L$ is a cubic cyclic
field extension
whose Galois group is generated by the automorphism $\tau: \omega_7 +\omega_7^{-1} \mapsto \omega_7^2 + \omega_7^{-2}$.
Note that $\omega\in\mathcal{O}_L=\mathbb{Z}[\omega]$.

The multiplication of the division algebra $A = (D, \tau, \omega)$
is behind the fully diverse
codes employed in
 \cite{R13} (cf. \cite{SP14}). Here,
$$\Lambda=\mathbb{Z}[\omega_3,\omega_7+\omega_7^{-1}]\oplus \mathbb{Z}[\omega_3,\omega_7+\omega_7^{-1}] e\oplus \mathbb{Z}[
\omega_3,\omega_7+\omega_7^{-1}]e^2\oplus\dots $$
is a natural order in $A$.

Let $p=2$, then $\mathfrak{p}$ is a prime ideal in $\mathcal{O}_F=\mathbb{Z}[\omega_3]$ which
 remains prime in $\mathcal{O}_K=\mathbb{Z}[\omega_3,\omega_7+\omega_7^{-1}]$ and
$\mathbb{Z}[i]/\mathfrak{p}\cong\mathbb{F}_4$.
 $\mathfrak{p}$ is inert in $K=\mathbb{Q}(\omega_3,\omega_7+\omega_7^{-1})$.
 Now
  $$\mathcal{D}/\mathfrak{p}\mathcal{D}=(\mathbb{F}_{64}/\mathbb{F}_8,\overline{\sigma},-1)\cong {\rm Mat}_2(\mathbb{F}_8)$$
 is a split quaternion algebra over $\mathbb{F}_8$.
Thus
$$\Lambda/\mathfrak{p}\Lambda\cong ({\rm Mat}_2(\mathbb{F}_8),\overline{\tau},\overline{\omega})$$
 where $\overline{\omega}\in \mathbb{Z}[\omega]/2\mathbb{Z}[\omega]$.
\end{example}

%*******************************************************************************************%
%
% Lattice encoding II
%
%*******************************************************************************************%

\section{Lattice encoding of cyclic $(f,\sigma,\delta)$-codes over $\mathcal{O}_K/\mathfrak{p}\mathcal{O}_K$, II}
\label{sec:codesII}

We continue to assume the setup from Section \ref{subsec:iteratedI}.

\subsection{A second generalization of Construction A}\label{sub:ex:nca}

Let $\mathcal{I}$ be a  principal left ideal of $\Lambda$ generated by a monic polynomial $g(t)$, such that
$\mathfrak{p}\subset \mathcal{I}\cap S_0$.
Then $\mathcal{I}/\mathfrak{p}\Lambda$ is a non-zero principal left ideal of $\Lambda/\mathfrak{p}\Lambda$ and
$\Psi(\mathcal{I}/\mathfrak{p}\Lambda)$
is a principal left ideal of $S_{\overline{f}}$ generated by the monic polynomial $\Psi(g+\mathfrak{p}\Lambda)=\overline{g}$.
That means $\Psi(\mathcal{I}/\mathfrak{p}\Lambda)$ corresponds to an
$(\overline{f},\overline{\sigma},\overline{\delta})$-code $\mathcal{C}$ over $\mathcal{O}_K/\mathfrak{p}\mathcal{O}_K$.

In particular, if we choose $\delta=0$ and $f(t)$ such that
$\overline{f}(t)=t^m-\overline{ c}\in \mathcal{D}/\mathfrak{p}\mathcal{D}$ with $\overline{c}$ non-zero, then $\Psi(\mathcal{I}/\mathfrak{p}\Lambda)$ is a
$\overline{\sigma}$-constacyclic code over $\mathcal{O}_K/\mathfrak{p}\mathcal{O}_K$.

If $\overline{f}$ is irreducible and $\mathcal{D}/\mathfrak{p}\mathcal{D}$ is a division algebra, then the algebra
$S_{\overline{f}}$ is simple. Then any non-zero code
$\mathcal{C}$ must have length $m$ and dimension $m$ (and correspond to the whole algebra),
whereas whenever $\overline{f}$ is reducible, $\mathcal{C}$ respectively $\Psi(\mathcal{I}/\mathfrak{p}\Lambda)$
 corresponds to a right divisor $\overline{g}$ of
$\overline{f}$ and has dimension  $m-{\rm deg}(\overline{g})$.
Let
$$\rho: \Lambda\longrightarrow \Lambda/\mathfrak{p}\Lambda \longrightarrow \Psi(\Lambda/\mathfrak{p}\Lambda)$$ be the
canonical projection $\Lambda\longrightarrow \Lambda/\mathfrak{p}\Lambda$ composed with $\Psi$.
We know that $\mathcal{O}_K$ is a free $\mathbb{Z}$-module of rank $n[F:\mathbb{Q}]$.
Therefore
$$L=\rho^{-1}(\mathcal{C})=\mathcal{I}$$
is a $\mathbb{Z}$-module of dimension $N=n^2m[F:\mathbb{Q}]$.
The embedding of this lattice into $\mathbb{R}^N$ is canonically determined by
 considering $S_f\otimes_{\mathbb{Q}}\mathbb{R}$. Again all works exactly as
 explained in \cite[Section 3.3]{DO} (since associativity is not relevant for the argument).
The construction of $L$ can again be seen as a second (nonassociative) variation of the
non-commutative Construction A in \cite{CS}.

In this way we can construct a lattice $L$ in $\mathbb{R}^N$ from the linear
$(\overline{f},\overline{\sigma},\overline{\delta})$-code $\mathcal{C}$ over a finite ring.

Note that even if $R=D[t;\sigma,\delta]$ is isomorphic to $D[t;\sigma]$ or $D[t;\delta]$, like when $\sigma$ or $\delta$
are inner, we conjecture that the codes/lattices we obtain from
using different ways to write $R$ can be substantially different in performance, similarly as the examples obtained in
\cite{BU14}, where some of the codes obtained by working with the general skew polynomial
ring $\mathbb{F}_q[t;\sigma,\delta]$  have a better distance bound than the ones obtained with
$\delta=0$.

\subsection{Space-time block codes}

We now apply the above considerations to space-time block coding (cf. \ref{subsec:STBC}).

\begin{example}
 Let $A=(K/F,\sigma,c)$, $c\in\mathcal{O}_K$ non-zero, be a nonassociative cyclic division algebra  over $F$
 of degree $m$ with $c\in \mathcal{O}_K$.
  Take the natural order $\Lambda=\mathcal{O}_K[t;\sigma]/\mathcal{O}_K[t;\sigma]f$, and
let $a = a_0 + a_1t + \cdots +  a_{m-1}t^{m-1},\,b = b_0 + b_1t + \cdots +  b_{m-1}t^{m-1}\in \Lambda$.
If we identify  $a $  with the vector $(a_0, a_1, \ldots, a_{m-1})$,
 we can express  $R_a\in {\rm End}_{\mathcal{O}_K}(\Lambda)$  as
 an $m \times m$-matrix $M(a)$    with entries in $\mathcal{O}_K$:
\begin{equation} \label{equ:matrix_rep_cda}
M(a)= \left[ \begin{array}{ccccc}
a_0 & c \sigma(a_{m-1})& c \sigma^2(a_{m-2}) & \cdots & c \sigma^{m-1}(a_1) \\
a_1 & \sigma(a_0) & c\sigma^2(a_{m-1}) & \cdots & c \sigma^{m-1}(a_{2}) \\
a_2 & \sigma(a_1) & \sigma^2(a_0) & \cdots & c \sigma^{m-1}(a_3)\\
\vdots & \vdots & \vdots & \ddots & \vdots \\
a_{m-1} & \sigma(a_{m-2}) & \sigma^2(a_{m-3}) & \cdots & \sigma^{m-1}(a_0) \end{array} \right]^T .
\end{equation}
If we identify $b$ with $(b_0, b_1, \ldots, b_{m-1})$, the right multiplication with $a$ in $\Lambda$ is given by
 the matrix multiplication $b\cdot a=bM(a) $.
The family of matrices $\{M(a)\,|\, 0\not=a\in A\}$ is a fully diverse  linear space-time
 block code $\mathcal{C}$.

 Let $\mathfrak{p}\subset \mathcal{O}_{F}$ be a maximal ideal. Then
$\rho: \Lambda\longrightarrow \Lambda/\mathfrak{p}\Lambda \longrightarrow \Psi(\Lambda/\mathfrak{p}\Lambda)=
 ((\mathcal{O}_K/\mathfrak{p}\mathcal{O}_K)/\mathbb{F}_{p^f},
\overline{\sigma},\overline{c}).$
Hence
$L=\rho^{-1}(\mathcal{C})=\mathcal{I}$
is a fully diverse space-time block code over $\mathcal{O}_K$ which
is a $\mathbb{Z}$-lattice whose  embedding into $\mathbb{R}^{m}$ is canonically determined by
 $A\otimes_{\mathbb{Q}}\mathbb{R}$.
\end{example}

Nonassociative cyclic division algebras as above can be employed to obtain fully diverse multiple-input double-output
codes \cite{SPO12}.
 The algebras $A=(D,\sigma^{-1},d)$ we consider next are used for the systematic space-time block code constructions
of the fast-decodable iterated codes in  \cite{MO13}, \cite{PS15},  \cite{P13.2},  \cite{R13}.

\begin{example}  Let $A=(D,\sigma,d)$ be a division algebra of degree $n$
and $d\in \mathcal{O}_L$ or $d\in \mathcal{O}_F$.

For $x=x_0+x_1t+x_2t^2+\cdots+x_{m-1}t^{m-1}$, $y=y_0+y_1t+y_2t^2+\cdots+y_{m-1}t^{m-1}\in \Lambda$,
$x_{i}, y_i \in \mathcal{D}$, represent $x$ as $(x_0, x_1, \ldots, x_{m-1})$ and $y$ as
$(y_0, y_1, \ldots, y_{m-1})$.
  Then $R_x\in {\rm End}_\mathcal{D}(\Lambda)$ is given by
the $m \times m$-matrix
\[M(x) = \left[ \begin{array}{ccccc}
x_0 & d \sigma(x_{m-1})& d \sigma^{2}(x_{m-2}) & \cdots & d \sigma^{m-1}(x_1) \\
x_1 & \sigma(x_0) & d \sigma^{2}(x_{m-1}) & \cdots & d \sigma^{m-1}(x_{2}) \\
x_2 & \sigma(x_1) & \sigma^{2}(x_0) & \cdots & d \sigma^{m-1}(x_3)\\
\vdots & \vdots & \vdots & \ddots & \vdots \\
x_{m-1} & \sigma(x_{m-2}) & \sigma^{2}(x_{m-3}) & \cdots & \sigma^{m-1}(x_0) \end{array} \right]^T \]
 with entries in $\mathcal{D}$.
We can write the multiplication in $\Lambda$  as $y\cdot x= yM(x)$.
Now substitute the right regular representation $\gamma(d)$ in $\mathcal{D}$ for $d$ in $M(x)$  and
 the right regular representation $\gamma(x_i)$  in $\mathcal{D}$ for each
entry $x_i$ in $M(x)$. This way we obtain a block matrix
\[\gamma(M(x)) := \left[ \begin{array}{cccc}
\gamma(x_0) & \gamma(d) \sigma(\gamma(x_{m-1}))&  \cdots & \gamma(d) \sigma^{m-1}(\gamma(x_1)) \\
\gamma(x_1) & \sigma(\gamma(x_0)) & \cdots & \gamma(d) \sigma^{m-1}(\gamma(x_{2})) \\
                                       \vdots & \vdots  & \ddots & \vdots \\
\gamma(x_{m-1}) & \sigma(\gamma(x_{m-2})) & \cdots & \sigma^{m-1}(\gamma(x_0)) \end{array} \right]^T
\]
where $\sigma(\gamma(x_i))$ means we apply $\sigma$ to each entry of the $n\times$-matrix $\gamma(x_i)$.
Products are the usual matrix products. This is an $mn \times mn$ matrix with entries in $\mathcal{O}_K$.
 It represents right multiplication in $\Lambda$.
 Writing elements in
$\Lambda=\mathcal{O}_K \oplus  \mathcal{O}_K e\oplus  \dots\oplus\mathcal{O}_K  e^{n-1}t^{m-1}$
as row vectors of length $mn$ with entries in $\mathcal{O}_K$, we obtain
$y\cdot x = y\, \gamma(M(x))$.

The family of matrices $\{ \gamma(M(x))\}$ is a fully diverse  linear space-time
 block code $\mathcal{C}$. In particular,
if $d\in \mathcal{O}_F$, then
$\det(\gamma(M(x))) \in \mathcal{O}_F$ (\cite{MO13}, \cite[Remark 5]{PS15})
and if $d\in \mathcal{O}_L$, then
\[ \gamma(M(x)) = \begin{bmatrix}
             \gamma(x_0) & d \sigma(\gamma(x_{n-1}))& d \sigma^{2}(\gamma(x_{n-2})) & \cdots & d \sigma^{m-1}(\gamma(x_1)) \\
             \gamma(x_1) & \sigma(\gamma(x_0)) & d \sigma^{2}(\gamma(x_{n-1})) & \cdots & d \sigma^{m-1}(\gamma(x_{2})) \\
             \gamma(x_2) & \sigma(\gamma(x_1)) & \sigma^{2}(\gamma(x_0)) & \cdots & d \sigma^{m-1}(\gamma(x_3))\\
             \vdots & \vdots & \vdots & \ddots & \vdots \\
             \gamma(x_{n-1}) & \sigma(\gamma(x_{n-2})) & \sigma^{2}(\gamma(x_{n-3})) & \cdots & \sigma^{m-1}(\gamma(x_0))
             \end{bmatrix}^T
             \]
and
$\det(\gamma(M(x))) \in L\cap \mathcal{O}_K=\mathcal{O}_L$ (\cite{R13}, \cite[Lemma 19]{PS15}).

 Let $\mathfrak{p}\subset \mathcal{O}_{F_0}$ be a prime ideal such that $\mathfrak{p}\mathcal{O}_F$ is maximal. Then
$$\rho: \Lambda\longrightarrow \Lambda/\mathfrak{p}\Lambda \longrightarrow \Psi(\Lambda/\mathfrak{p}\Lambda)
=(\overline{D}, \overline{\sigma}, \overline{d}).$$
Here,
$L=\rho^{-1}(\mathcal{C})=\mathcal{I}$
induces a fully diverse STBC over $\mathcal{O}_K$ which
is a $\mathbb{Z}$-lattice whose  embedding into $\mathbb{R}^N$, $N=mn^2$, is canonically determined by
 $A\otimes_{\mathbb{Q}}\mathbb{R}$.
\end{example}

\begin{remark}
The explanations in \cite[Section 5.2, 5.3]{DO} hold analogously for our generalization of Construction A
in Section \ref{ex:nca} and the examples here, and
 show the potential of the construction for coset coding used in space-time block coding, in particular for wiretap
 space-time block coding, but also for linear codes over finite rings.
Moreover, the matrix generating a cyclic $(f,\sigma,\delta)$-code $\mathcal{C}\subset S^m$
 represents the right multiplication $R_g$ in $S_f$ and is a control matrix of $\mathcal{C}$  \cite{P15}.
\end{remark}

%*******************************************************************************************%
%
% Conclusion
%
%*******************************************************************************************%

\section{Conclusion}

We presented a method how to construct a lattice from a suitable $(f,\sigma,\delta)$-code defined over a finite ring
which can be seen as a generalization of the classical construction A.
This can be summarized as follows:
Let $D$ be a cyclic division algebra over $F$ which is already defined over $\mathcal{O}_F$, or a Galois field extension and
 $f$ defined over its ring of integers.
Take the additional assumptions on $\sigma$ and $\delta$ as given in the corresponding previous sections.
\begin{itemize}
\item Choose some monic skew polynomial $f\in \mathcal{D}[t;\sigma,\delta]$
(resp., $f\in \mathcal{O}_K[t;\sigma,\delta]$ in the field case) which is irreducible in $D[t;\sigma,\delta]$.
\item Take a natural order $\Lambda$ of $S_f$.
\item Choose a prime ideal $\mathfrak{p}$ in $S_0$. This yields the finite ring
$\mathcal{O}_K/\mathfrak{p}\mathcal{O}_K$
  you consider the code $\mathcal{C}$  to be defined over.
 $\overline{f}$ must be reducible in $(\mathcal{D}/\mathfrak{p}\mathcal{D})[t;\overline{\sigma},\overline{\delta}]$.
\item Choose a principal left ideal $\mathcal{I}$  of $\Lambda$ generated by a monic polynomial, such that
$\mathfrak{p}\subset \mathcal{I}\cap S_0$.
\item
 $\Psi(\mathcal{I}/\mathfrak{p}\Lambda)$ corresponds to an
$(\overline{f},\overline{\sigma},\overline{\delta})$-code $\mathcal{C}$ over $\mathcal{O}_K/\mathfrak{p}\mathcal{O}_K$, and
$L=\rho^{-1}(\mathcal{C})=\mathcal{I}$
is a $\mathbb{Z}$-lattice whose embedding into $\mathbb{R}^N$ is canonically determined by
 $S_f\otimes_{\mathbb{Q}}\mathbb{R}$.
\end{itemize}

If we want to apply this construction to space time block coding instead, we substitute the last step with:
\begin{itemize}
\item Take the matrix representing right multiplication in $\Lambda$ and let $\mathcal{C}$ be the
associated space-time block code. Then $\rho^{-1}(\mathcal{C})=\mathcal{I}$ is a fully diverse space-time block code
which is a $\mathbb{Z}$-lattice.
\end{itemize}

If desired, this method can be extended to work for any Noetherian integral domain and central simple algebra  over its
quotient field.
It can be applied for coset coding and wiretap coding analogously as described in \cite[Sections 5.2, 5.3]{DO}.

It would be interesting to investigate which properties of $\mathcal{C}$ carry over to the lattice STBC $L$
and find examples of well performing coset codes.

%*******************************************************************************************%
%****************************************************************************************%


\begin{thebibliography}{1}

\bibitem{Am} A. S. Amitsur, \emph{Differential polynomials and division algebras.}
 Annals of Mathematics, Vol. 59  (2) (1954) 245-278.

\bibitem{Am2} A. S. Amitsur, \emph{Non-commutative cyclic fields.}
Duke Math. J. 21 (1954), 87–105.

\bibitem{M0} J. Apel, \emph{Gr\"obnerbasen in Nichtkommutativen Algebren und ihre Anwendung}.
Dissertation, Leipzig (1988)


\bibitem{A} M. Artin, \emph{Noncommutative rings.}
http://math.mit.edu/~etingof/artinnotes.pdf, last accessed 24.4.2017


\bibitem{BZ} J. Bergen, M. Giesbrecht,  P. N. Shivakumar, Y. Zhang,
\emph{Factorizations for difference operators}. Adv. Difference Equ. 2015, 2015:57, 6 pp.


\bibitem{B} M. Bhaintwal,  \emph{Skew quasi-cyclic codes over Galois rings}. Des. Codes Cryptogr. 62 (1) (2012), 85-101.

\bibitem{BG}  A. Batoul, K. Guenda, T. A. Gulliver,  \emph{On self-dual cyclic codes over finite chain rings}.
  Des. Codes Cryptogr. 70 (3) (2014),  347-358.



\bibitem{BSU08} D. Boucher, P. Sol\`e, F. Ulmer, \emph{Skew-constacyclic codes over Galois rings}.
 Adv. Math. Comm. 2 (3) (2008), 273-292.

\bibitem{BU14} D. Boucher, F. Ulmer, \emph{Linear codes using skew polynomials with automorphisms and derivations},
 Des. Codes Cryptogr. 70 (3) (2014), 405-431.

\bibitem{BU14.2} D. Boucher, F. Ulmer,
\emph{Self-dual skew codes and factorization of skew polynomials}, J. Symbolic Comput. 60 (2014), 47-61.

\bibitem{BU09.2} D. Boucher, F. Ulmer, \emph{Coding with skew polynomial rings},
J. Symbolic Comput. 44 (12) (2009),  1644-1656.

\bibitem{BU09} D. Boucher, F. Ulmer,
\emph{Codes as modules over skew polynomial rings. Cryptography and coding},
Lecture Notes in Comput. Sci., 5921, Springer, Berlin, 2009, 38-55.

\bibitem{BGU07} D. Boucher, W. Geiselmann, F. Ulmer,  \emph{Skew-cyclic codes.} AAECC 18 (2007), 370-389.


\bibitem{BL13} M. Boulagouaz, A. Leroy, \emph{$(\sigma,\delta)$-codes.}  Adv. Math. Comm. 7 (4) (2013), 463-474.

 \bibitem{BGV}
J. Bueso, J. Gomez-Torrecillas, and A. Verschoren. ``Methods in Non-Commutative Algebra.'' (2003). Kluwer
Academic Press.

 \bibitem{Cao}   Y. Cao, \emph{On constacyclic codes over finite chain rings.} Finite Fields Appl. 24 (2013), 124-135



 \bibitem{CM} M. Ceria, T. Mora, \emph{Buchberger–Zacharias Theory of multivariate Ore extensions.}
  Journal of Pure and Applied Algebra
221 (12)  (2017), 2974-3026.

\bibitem{C} P. M. Cohn,
``Skew fields.'' Theory of general division rings.
Encyclopedia of Mathematics and its Applications, 57. Cambridge University Press, Cambridge, 1995.


\bibitem{CS} J. H. Conway, N. J. Sloane, \emph{Sphere packings, Lattices and groups.} Springer Verlag 1999.

\bibitem{DO.0} J.~Ducoat, F.~Oggier, \emph{Lattice encoding of cyclic codes from skew polynomial rings}.
Proc. of the 4th International Castle Meeting on Coding Theory and Applications, Palmela, 2014.


\bibitem{DO} J.~Ducoat, F.~Oggier, \emph{On skew polynomial codes and lattices from quotients of
 cyclic division algebras.} Adv. Math. Comm. 10 (1) 2016, 79-94.

\bibitem{FG} N. Fogarty, H. Gluesing-Luerssen, \emph{A Circulant Approach to Skew-Constacyclic Codes.}
 Finite Fields Appl. 35 (2015), 92–114.

\bibitem{G1} J.~G\'{o}mez-Torrecillas, F. J.~Lobillo, G.~Navarro, \emph{A new perspective of cyclicity in convolutional codes.}
 IEEE Trans. Inform. Theory 62 (5) (2016),  2702-2706.

\bibitem{G2} J.~G\'{o}mez-Torrecillas, F. J.~Lobillo, G.~Navarro, \emph{Convolutional codes with a matrix-algebra word-ambient.}
 Adv. Math. Commun. 10 (1) (2016),  29-43.

\bibitem{G3} J.~G\'{o}mez-Torrecillas, F. J.~Lobillo, G.~Navarro, \emph{An isomorphism test for modules over a non-commutative
 PID. Applications to similarity of Ore polynomials.} J. Symbolic Comput. 75 (2016), 149-170.

\bibitem{G4} J.~G\'{o}mez-Torrecillas, F. J.~Lobillo, G.~Navarro,
\emph{Separable automorphisms on matrix algebras over finite field extensions: applications to ideal codes.}
 ISSAC'15 - Proceedings of the 2015 ACM International Symposium on Symbolic and Algebraic Computation,
 189–195, ACM, New York, 2015.

\bibitem{G5} J.~G\'{o}mez-Torrecillas, F. J.~Lobillo, G.~Navarro, \emph{Information-bit error rate and false positives
in an MDS code.} Adv. Math. Commun. 9 (2) (2015), 149-168.

 \bibitem{G6} J.~G\'{o}mez-Torrecillas, \emph{ Basic module theory over non-commutative rings with computational aspects
  of operator algebras. With an appendix by V. Levandovskyy.} Lecture Notes in Comput. Sci. 8372, Algebraic and
  algorithmic aspects of differential and integral operators, 23-82, Springer, Heidelberg, 2014.


 \bibitem{GK}   J.  Gao, Kong, \emph{Qiong 1-generator quasi-cyclic codes over
    $\mathbb{F}_{p^m}+u\mathbb{F}_{p^m}+\dots+u^{s-1}\mathbb{F}_{p^m}$.} J. Franklin Inst. 350 (10) (2013), 3260-3276.



\bibitem{Hoe} Hoechsmann, Klaus, \emph{Simple algebras and derivations.} Trans. Amer. Math. Soc. 108 (1963), 1-12.


\bibitem{J96} N.~Jacobson,
``Finite-dimensional division algebras over fields,'' Springer Verlag,
Berlin-Heidelberg-New York, 1996.



 \bibitem{Ka} J. S. Kauta,  \emph{Maximal orders and valuation rings in nonassociative quaternion algebras.}
 Proceedings of the 39th Symposium on Ring Theory and Representation Theory, 65-74, Symp. Ring Theory Represent.
 Theory Organ. Comm., Yamaguchi, 2007.

 \bibitem{KrW} A. Kandri-Rody, W. Weispfenning, \emph{Non-commutative Gr\"obner bases in algebras of solvable type}. J. Symb. Comp. 9 (1990), 1-26


  \bibitem{Kr} H. Kredel, \emph{Solvable polynomial rings.}
Shaker Verlag 1993.

\bibitem{LW} H. J. Lee, W. C. Waterhouse,   \emph{Maximal orders in nonassociative quaternion algebras.}
 J. Algebra 146 (2) (1992), 441-453.

 \bibitem{Lev} V. Levandovskyy, \emph{Non-commutative computer algebra for polynomial algebras:
Gr\"obner bases, applications and implementation.} Dissertation,
 Kaiserslautern  (2005),
{\tt http://kluedo.ub.uni-kl.de/volltexte/2005/1883/}


    \bibitem{Lev2} V. Levandovskyy, \emph{PBW bases, non-degeneracy conditions and applications.}
  In: Buchweitz, R.-O., Lenzing, H. (Eds.),
Representation of Algebras and Related Topics. Proceedings of
 the ICRA X Conference,
 45 (2005), 229-246.





\bibitem{LL}  X. Liu, H. Liu, \emph{LCD codes over finite chain rings}. Finite Fields Appl. 34 (2015), 1-19.




\bibitem{MO13} N.~Markin, F.~Oggier, \emph{Iterated space-time code constructions from cyclic algebras.}
  IEEE Transactions on Information Theory,  59 (9), September 2013.

\bibitem{N1} G.~Nebe, A.~Schaefer, \emph{A nilpotent non abelian group code.}
Algebra and Discrete Mathematics 18 (2014) 268-273.

\bibitem{N2} G.~Nebe, W.~Willems, \emph{On self-dual MRD codes.}
Adv.  Math. of Comm. 10 (2016) 633-642.


 \bibitem{N3} G. Nebe, D. Liebhold and A. Vazquez Castro,
\emph{Network coding with flags.} To appear in Designs, Codes and Cryptography.




\bibitem{OS} F.~Oggier, B.~A.~Sethuraman, \emph{Quotients of orders in cyclic algebras and space-time codes}.
 Adv. Math. Commun. 7 (4) (2013),  441-461.

 \bibitem{OB} F.~Oggier, J.-C. Belfiore, \emph{Enabling multiplication in lattice codes via Construction A}.
IEEE Int. Workshop Inf. Theory 2013, 1-5.


\bibitem{O} O. Ore, \emph{Formale Theorie der linearen Differentialgleichungen. (Zweiter Teil)}. (German)
 J. Reine Angew. Math. 168 (1932), 233-252.

 \bibitem{O1} O. Ore, \emph{Theory of noncommutative polynomials.} Annals of Math. (1933), 480-508.

  \bibitem{P1} M. Pesch, \emph{ Gr\"obner Bases in skew polynomial rings.}
Shaker Verlag 1998.


     \bibitem{P2} M. Pesch, \emph{Two-sided Gr\"obner bases in iterated Ore extensions.}
 Progress in Computer Science and Applied Logic (15) (1991), 225-243, Birkh\"auser.


\bibitem{P66} J.-C. Petit, \emph{Sur certains quasi-corps g\'{e}n\'{e}ralisant un type d'anneau-quotient},
S\'{e}minaire Dubriel. Alg\`{e}bre et th\'{e}orie des nombres 20 (1966-67), 1-18.

\bibitem{P68} J.-C. Petit, \emph{Sur les quasi-corps distributifes \`{a} base momog\`{e}ne},
C. R. Acad. Sc. Paris  266 (1968), S\'{e}rie A, 402-404.


\bibitem{SP14} S. Pumpl\"un, A. Steele,
\emph{Fast-decodable MIDO codes from nonassociative algebras.}
 Int. J. of Information and Coding Theory (IJICOT) 3 (1) 2015, 15-38.


\bibitem{PS15} S.~Pumpl\"un, A. Steele, \emph{The nonassociative algebras used to build fast-decodable space-time block
codes}.  Adv.  Math.  Comm. 9 (4) (2015), 449-469.

 \bibitem{P13.2}  S.~Pumpl\"un, \emph{How to obtain division algebras used for fast decodable space-time block codes}.
  Adv.  Math.  Comm. 8 (3) (2014), 323 - 342.


\bibitem{P15} S.~Pumpl\"un, \emph{Finite nonassociative algebras obtained from skew polynomials
and possible applications to $(f,\sigma,\delta)$-codes}.
 Adv.  Math.  Comm.
11 (3) (2017), 615-634. doi:10.3934/amc.2017046



\bibitem{S} R. Sandler,  \emph{Autotopism groups of some finite non-associative algebras}.
 American Journal of Mathematics 84 (1962), 239-264.


\bibitem{Sch} R.D. Schafer, ``An introduction to nonassociative algebras,'' Dover Publ., Inc., New York, 1995.


\bibitem{R13} K. P.~Srinath, B. S.~Rajan, \emph{Fast-decodable MIDO codes with large coding gain.}
{\em IEEE Transactions on Information Theory} (2) 60  2014, 992-1007.

\bibitem{SPO12} A.~Steele, S.~Pumpl\"un, F.~Oggier,
\emph{MIDO space-time codes from associative and non-associative cyclic algebras.}
 Information Theory Workshop (ITW) 2012 IEEE (2012), 192-196.


\bibitem{S12} A.~Steele, \emph{Nonassociative cyclic algebras.}
 Israel J. Math. 200 (1) (2014),  361-387.


\bibitem{W} W. C. Waterhouse, {\it Nonassociative quaternion algebras}. Algebras, Groups
and Geometries 4 (1987), 365-378.


 \bibitem{Wei} V. Weispfenning, \emph{Finite Gr\"obner bases in non-noetherian skew polynomial rings.}
 Proc. ISSAC'92 (1992), 320-332, A.C.M.


\bibitem{MW}   M. Wu,  \emph{Free cyclic codes as invariant submodules over  finite chain rings}.
Int. Math. Forum 8 (37-40) (2013), 1835-1838.

\end{thebibliography}
\end{document}